\documentclass[journal,table,xcdraw]{IEEEtran}
\usepackage[table,xcdraw]{xcolor}
\usepackage{graphicx}
\graphicspath{../jpeg/}
\DeclareGraphicsExtensions{.pdf,.jpeg,.png}
\usepackage[cmex10]{amsmath}
\allowdisplaybreaks[4]
\usepackage{cuted}
\usepackage{stfloats}%
\usepackage{amssymb}
\usepackage{array}
\usepackage{enumitem}
\usepackage{mdwmath}
\usepackage{mdwtab}
\usepackage{eqparbox}
\usepackage{url}
\usepackage{enumerate}
\usepackage{amsfonts}
\usepackage{algorithmic}
\usepackage{multirow}
\usepackage{makecell}
\usepackage{mathtools}
\usepackage[ruled,vlined,linesnumbered]{algorithm2e}
\usepackage{nomencl}
\usepackage{cite}
\usepackage{textcomp,booktabs}
\usepackage{hyperref}

\usepackage{mathrsfs}
\usepackage{amsthm}
\theoremstyle{remark}

\newtheorem{theorem}{Theorem}
\newtheorem{lemma}{Lemma}

\newtheorem{prop}{Proposition}

\SetKwInput{KwInput}{Input}                
\SetKwInput{KwOutput}{Output} 
\makeatletter
\def\underbracex#1#2{\mathop{\vtop{\m@th\ialign{##\crcr
        $\hfil\displaystyle{#2}\hfil$\crcr
        \noalign{\kern3\p@\nointerlineskip}%
        #1\crcr\noalign{\kern3\p@}}}}\limits}

\def\upbracefilla{$\m@th \setbox\z@\hbox{$\braceld$}%
\bracelu\leaders\vrule \@height\ht\z@ \@depth\z@\hfill 
\kern\p@\vrule \@width\p@\kern\p@\vrule \@width\p@\kern\p@\vrule \@width\p@
$}

\def\upbracefillb{$\m@th \setbox\z@\hbox{$\braceld$}%
\vrule \@width\p@\kern\p@\vrule \@width\p@\kern\p@\vrule \@width\p@\kern\p@
\leaders\vrule \@height\ht\z@ \@depth\z@\hfill\bracerd
\braceld\leaders\vrule \@height\ht\z@ \@depth\z@\hfill
\kern\p@\vrule \@width\p@\kern\p@\vrule \@width\p@\kern\p@\vrule \@width\p@
$}

\def\upbracefillc{$\m@th \setbox\z@\hbox{$\braceld$}%
\vrule \@width\p@\kern\p@\vrule \@width\p@\kern\p@\vrule \@width\p@\kern\p@
\leaders\vrule \@height\ht\z@ \@depth\z@\hfill
\kern\p@\vrule \@width\p@\kern\p@\vrule \@width\p@\kern\p@\vrule \@width\p@
$}

\def\upbracefilld{$\m@th \setbox\z@\hbox{$\braceld$}%
\vrule \@width\p@\kern\p@\vrule \@width\p@\kern\p@\vrule \@width\p@\kern\p@
\leaders\vrule \@height\ht\z@ \@depth\z@\hfill\braceru$}

\def\upbracefillbd{$\m@th \setbox\z@\hbox{$\braceld$}%
\vrule \@width\p@\kern\p@\vrule \@width\p@\kern\p@\vrule \@width\p@\kern\p@
\bracerd\braceld
\leaders\vrule \@height\ht\z@ \@depth\z@\hfill\braceru$}

\begin{document}

\title{Hybrid Oscillation Damping and Inertia Management for Distributed Energy Resources}

\author{Cheng Feng,~\IEEEmembership{Member,~IEEE,}
Linbin Huang,~\IEEEmembership{Member,~IEEE,}
Xiuqiang He,~\IEEEmembership{Member,~IEEE,}\\
Yi Wang,~\IEEEmembership{Member,~IEEE,}	
Florian Dörfler,~\IEEEmembership{Senior Member,~IEEE,}
Chongqing Kang,~\IEEEmembership{Fellow,~IEEE}.

\thanks{
Manuscript received 28 August 2024; revised 10 Feb 2025; accepted 16 Apr 2025. This work is supported by the National Key Research and Development Program of China under Grant 2021YFB2401200 and the National Natural and Science Foundation of China under Grant 52321004. Paper no. TPWRS-01440-2024. (\textit{Corresponding author: Chongqing Kang}.)

Cheng Feng and Chongqing Kang are with the Department of Electrical Engineering, Tsinghua University, Beijing 100084, China (email: fengc.2019@tsinghua.org.cn; cqkang@tsinghua.edu.cn). 

Linbin Huang is with the College of Electrical Engineering, Zhejiang University, Hangzhou 310027, China (email: hlinbin@zju.edu.cn). 

Xiuqiang He is with the Department of Automation, Tsinghua University, Beijing 100084, China (email: hxq19@tsinghua.org.cn). 

Yi Wang is with the Department of Electrical and Electronic Engineering, The University of Hong Kong, Hong Kong SAR, China (email: yiwang@eee.hku.hk). 

Florian Dörfler is with Automatic Control Laboratory, ETH Zurich, 8092 Zurich, Switzerland (email: dorfler@ethz.ch). }}

\markboth{ }
{Shell \MakeLowercase{\textit{et al.}}: Bare Demo of IEEEtran.cls for IEEE Journals}
\maketitle

\begin{abstract}
Power systems dominated by converter-interfaced distributed energy resources (DERs) typically exhibit weaker damping capabilities and lower inertia, compromising system stability. Although individual DER controllers are evolving to provide superior oscillation damping capabilities and inertia supports, there is a lack of network-wide coordinated management measures for multiple DERs, potentially leading to unexpected instability and cost-effectiveness problems. To address this gap, this paper introduces a hybrid oscillation damping and inertia management strategy for multiple DERs, considering network coupling effects, and seeks to encourage DERs to provide enhanced damping and inertia with appropriate economic incentives. We first formulate an optimization problem to tune and allocate damping and inertia coefficients for DERs, minimizing associated power and energy costs while ensuring hard constraints for system frequency stability and small-signal stability. The problem is built upon a novel convex parametric formulation that integrates oscillation mode location and frequency trajectory requirements, equipped with a theoretical guarantee, and eliminating the need for iterative tuning and computation burdens. Furthermore, to increase the willingness of DERs to cooperate, we further design appropriate economic incentives to compensate for DERs' costs based on the proposed cost minimization problem, and assess its impact on system cost-efficiency. Numerical tests highlight the effectiveness of the proposed method in promoting system stability and offer insights into potential economic benefits.
\end{abstract}


\section{Introduction}
The power system is currently undergoing a profound transformation, with an increasing share of distributed energy resources (DERs) connecting to the grid via power electronic converters. This shift has generally resulted in lower inertia and weakened damping~\cite{Markovic2021,F.F.-1}. To react with a faster frequency response and to dampen an oscillation component swiftly, various studies aimed to design new grid-forming control structures with advanced inertia support and damping capabilities for \textit{individual} DERs~\cite{A.D.-5}.

Meanwhile, there is a lack of \textit{network-wide} management strategies for multiple DERs to provide oscillation damping and inertia, leading to instability and cost-effectiveness issues. On one hand, DERs' control parameters are often tuned individually, and network-wide interactions are overlooked. This oversight can result in unexpected stability issues when different DERs are connected via power networks~\cite{XiongLiu-16}. On the other hand, to provide oscillation damping and inertia support capabilities, DERs need to deload from the maximum available power point or use energy storage, incurring additional costs. Uncoordinated and ad hoc management of DERs' oscillation damping and inertia increases system costs~\cite{dreidy2017inertia}. Moreover, current system security management relies on the compulsory reliability criterion without compensations. These minimum-requirement reliability tests for DERs risk a `tragedy of the commons' scenario: where each DER adheres only to the basic stability requirements, cumulatively edging the entire system towards potential instability~\cite{ChengFan-17}. 
\subsection{Literature Review}
Conventionally, generation unit controller parameters are tuned individually by placing the system mode, using either static~\cite{chow1989pole} or adaptive controllers~\cite{cespedes2014adaptive}. It was generally assumed that these generation units were connected to an infinite bus through network inductance during tuning. However, this idealized model may not accurately reflect multiple DERs' interactions. As seen, unexpected instability issues will arise when DERs are connected to the grid~\cite{cheng2022real}. Impedance-based methods have been widely used to analyze the coupling effects among DERs~\cite{yang2021siso}. Ref.\cite{C.L.-22} suggested that different placement strategies of grid-forming DERs at different network nodes will affect the system dynamic performance. From the system operator's perspective, it is challenging to determine the exact impedance parameter for each DER. A more practical solution is to coordinately tune the most critical parameters: the oscillation damping and inertia coefficients~\cite{wang2022analytical}.

Network-wide stability-oriented unit management strategies have been a critical concern since the era of systems dominated by synchronous machines~\cite{ChenPan-18}. Various studies aimed to determine the limits of transient transmission power~\cite{GanThomas-21,XuDong-19}, or iteratively re-dispatch synchronous generators to maintain system stability~\cite{R.F.-15}. These studies explored ways to include transient stability~\cite{GanThomas-21,XuDong-19} and voltage stability constraints~\cite{CuiSun-20} into synchronous unit commitment models. The integration of DERs has further sparked research on this topic. For DERs, control parameters can be directly tuned to improve system stability. During real-time operations, Ref.~\cite{A.U.-26,L.Z.-27} explored the direct adjustment of the damping coefficients of DERs to promote system response. However, current network-wide converter tuning strategies are often built upon non-convex optimization models. They employed sensitivity factors of system modes to iteratively adjust the modal damping ratio, leading to significant computational complexity. Some research has focused on analytically calculating oscillation modes~\cite{GuoZhao-28,jiang2020dynamic} or time-domain responses~\cite{SajadiKenyon-30,M.X.-29}. These works assumed that all units have equal inertia-to-damping ratios, whereas the ratios can vary in real systems. 

Another important consideration of network-wide oscillation damping and inertia management for DERs is the associated economic incentives and costs. Minimizing related costs while encouraging DERs to actively provide superior dynamic performances is the key to a cost-efficient and secure DER integration. New grid codes have already been established to economically incentivize DERs to actively contribute to enhanced dynamic system performance beyond basic requirements. Notable examples include the inertia markets established by the UK National Grid \cite{UK} and the Australian Energy Market Operator (AEMO) \cite{AEMC}. Most related studies focused on the perspective of frequency stability and inertia. Ref.\cite{L.F.-23,Z.R.-24} investigated the cost-efficient provision of inertia and fast frequency reserves from DERs to improve transient performance. Ref.\cite{B.D.-3} further highlighted the strategic placement of inertia across DERs to enhance key dynamic performance metrics. Meanwhile, some work suggests that, in systems with a high penetration of DERs, small-signal stability is of significant concern~\cite{Markovic2021}, and damping plays a more pivotal role than inertia \cite{SajadiKenyon-30,li2023wholesystem}. 
\subsection{Contribution and Structure}
In summary, this paper proposes a \textit{hybrid oscillation damping and inertia management strategy} for DERs from the \textit{network-wide} perspective. The contributions are as follows:
\begin{itemize}[leftmargin=*]
\item \textit{Management problem formulation}: Develop an optimization problem for hybrid oscillation damping and inertia management for DERs. The problem minimizes related costs to place damping and inertia throughout the network, ensuring both frequency stability and small-signal stability while achieving cost-efficiency. 
\item \textit{Convex solution method}: Formulate convex constraints to ensure satisfactory oscillation damping and fast frequency response performance considering network effects. These constraints eliminate the need for iterative adjustments among DERs during tuning and rely on minimal assumptions. Uncertain load conditions and system parameters can also be integrated into the constraints using robust optimization techniques, while preserving convexity.
\item \textit{Practical incentive design}: Design economic incentives for hybrid oscillation damping and inertia managements through ancillary service compensations. Numerical results demonstrate the effectiveness of the proposed strategy in ensuring system stability and improving cost-effectiveness.
\end{itemize}

In Section \ref{optimization}, the management framework is explained and the models for different types of generation units are provided. The section also elaborates on the convex stability constraints related to damping and virtual inertia. In Section \ref{VCG}, we analyze the stability-related costs of DERs, and present the complete cost minimization model to allocate damping and inertia. This section also details the incorporation of economic incentives via ancillary services. Section \ref{case} offers a numerical validation of the proposed service and its economic implications. Conclusions are provided in Section \ref{conclusions}.

\textit{Notation:} Vectors (column vectors by default) are represented in bold italics. Matrices are represented in bold upright font. $\mathbf{I}$ indicates the identity matrix and $\mathbf{1}$ denotes an all-one vector. Transpose is represented by $\boldsymbol{x}^{\top}$. Conjugate is represented by $\boldsymbol{x}^{*}$. Conjugate transpose is represented by $\boldsymbol{x}^{\mathrm{H}}$. The inner product $\boldsymbol{x}$ and $\boldsymbol{y}$ is denoted as $\left< \boldsymbol{x}, \boldsymbol{y} \right>=\boldsymbol{x}^{\mathrm{H}}\boldsymbol{y}$. $h(s)$ refers to the Laplace domain representation of a function $h(t)$, while the latter refers to its time domain representation. $\left\| x(t) \right\| _{\infty}=\max_t |x(t)|$ represents the $\mathcal{L}_{\infty}$ norm. $\mathrm{col}(\boldsymbol{x}_1,\boldsymbol{x}_2,..,)$ denotes the stacked column vector formed by $\boldsymbol{x}_1,\boldsymbol{x}_2,...$. $\nabla$ is the gradient operator.  

\section{System Model and Stability Constraints}\label{optimization}
\subsection{Setting and System Model}
\begin{figure}[t]
\centering
\includegraphics[width=0.49\textwidth]{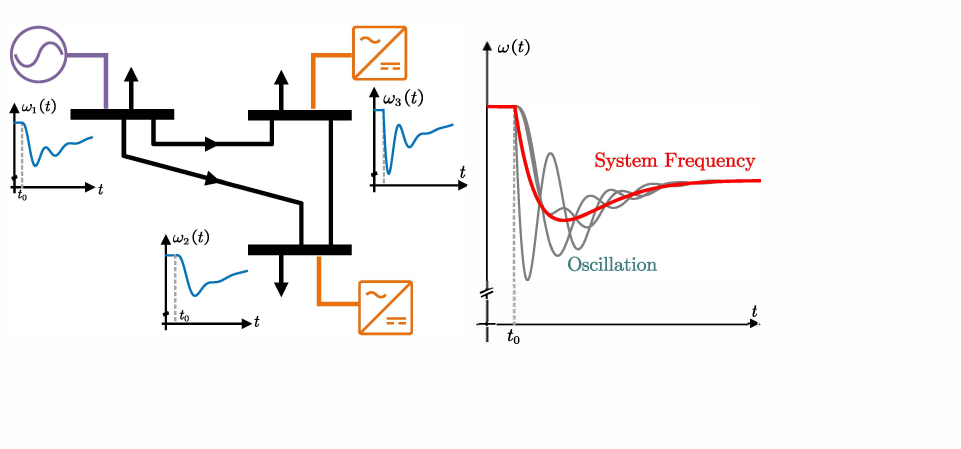}
\caption{The hybrid oscillation damping and inertia provision management aims to mitigate oscillations and improve poor frequency performances via tune and allocate nodal damping $d$ and virtual inertia $m$ from DERs.}
\label{fig:intro}
\end{figure}
The primary goal of the power system operator is to maintain system stability with minimum costs, ensuring that, even after significant disturbances, transient response performance remains within acceptable limits. It is crucial to ensure that the damping and inertia levels at different nodes can effectively mitigate the adverse effects of disturbances, including inter-area oscillations and poor frequency response performance. 

Consider the system as depicted in Fig.\ref{fig:intro}, where the generation units (including the DERs and the synchronous generators) are connected via a lossless transmission grid. In this setup, the buses are represented by the set $\mathcal{N}=\{1,..,i,..,I\}$. The set of buses $\mathcal{N}$ is partitioned into $\mathcal{N}_{\mathrm{G}}=\{1,...,I_1\}$ and $\mathcal{N}_{\mathrm{L}}=\{I_1+1,...,I\}$. $\mathcal{N}_{\mathrm{G}}$ denotes the set of buses connected to generation units (if any), while $\mathcal{N}_{\mathrm{L}}$ refers to those connected exclusively to loads. We focus on the dynamics of the voltage phasor angle and frequency for each bus $i$. The unified algebraic and differential equations describing their dynamics are~\cite{kundur2007power,Peter}: 
\begin{align}
\dot{\theta}_i=&~\omega _i,
\\
m_i\dot{\omega}_i=&-d_i\omega _i+P_{i}^{\mathrm{mec}}-P_{i}^{\mathrm{ld}}-\sum_{j\in \mathcal{N} _i}^{}{P_{i,j}},\forall i\in \mathcal{N} _{\mathrm{G}},
\label{eq:converter}
\\
0=&-P_{i}^{\mathrm{ld}}-\sum_{j\in\mathcal{N}_i}{P_{i,j}}~~~~~~~~~~~~~~~~~~~,\forall i\in \mathcal{N} _{\mathrm{L}}, \label{eq:load}
\end{align}
where $P_{i}^{\mathrm{mec}}$ and $P_{i}^{\mathrm{ld}}$ are the mechanical power input and electrical load power respectively; $j \in \mathcal{N}_i$ is the buses connected to bus $i$. The term $P_{i,j}$ represents the power flow from bus $i$ to $j$, where $j \in \mathcal{N}_i$ signifies the set of buses connected to $i$. The sum $\sum\nolimits_{j\in\mathcal{N}_i}{P_{i,j}}$ indicates the total electrical power injected into the grid via bus $i$. $\theta _i$ and ${\omega}_i$ denote the voltage phase angle and frequency of the bus $i$, respectively. All state variables represent deviations from their nominal values. 

If the bus $i$ is connected to a synchronous generator, Eq.\eqref{eq:converter} describes the electromechanical swing dynamics of the generator rotor angle $\theta_i$ and the rotational frequency $\omega_i$. Here, $m_i$ denotes the rotating inertia. $d_i$ is the equivalent damping provided via the amortisseur winding and related controllers. $P_{i}^{\mathrm{mec}}$ is the mechanical power input variation. The inertia and damping of synchronous machines are influenced by their inherent mechanical inertia, excitation, and windings, and are generally considered non-adjustable variables during real-time operations.

If $i$ is a DER, $m_i$ is the emulated inertia. $d_i$ is the frequency damping typically provided via droop coefficients~\cite{Lasseter2020}, and $P_{i}^{\mathrm{mec}}=0$. For the grid-forming (GFM) DERs, the control law $(m_i^{\mathrm{GFM}} s + d_i^{\mathrm{GFM}}) \Delta \omega_i = -\Delta P_i$ governing the synchronous loop directly gives the corresponding damping and inertia coefficients. $m_i^{\mathrm{GFM}}$ represents the equivalent inertia constant and $d_i^{\mathrm{GFM}}$ is the damping coefficient. For grid-following (GFL) DERs, the inertia effects originate from the phase-locked loop (PLL) dynamics. To provide damping capabilities, GFL DERs are assumed to utilize droop control in their active power control loop, expressed as $ \Delta P_i = -{d_i} \Delta \omega_i^{\mathrm{c}}$, where $\omega_i^{\mathrm{c}}$ denotes the PLL-estimated frequency and \( d_i \) represents the value of the inverse droop coefficient. It has been shown that $m_i^{\mathrm{GFL}} = d_i  \frac{\zeta_i}{\omega_{\mathrm{PLL}_i}}$ and $m_i^{\mathrm{GFL}} = m_i $, where \( \zeta_i \) is the PLL damping ratio and \( \omega_{\mathrm{PLL}_i} \) is the PLL bandwidth \cite{ducoin2024analytical}. The model uses reduced dynamics for DERs and assumes time-scale separation between outer and inner loops \cite{yang2022small,nandanoori2020distributed}. In doing so, we avoid compromising DER manufacturers' technical privacy and reduce the model's dimensionality.

When a bus hosts multiple generation units ($k=1,...,K_i$), $m_i$, $d_i$, and $P_{i}^{\mathrm{mec}}$ are the aggregated sums of the respective attributes of these units, representing their combined behavior. 

We denote the power injection to the grid via bus $i$ as $P_i = \sum\nolimits_{j\in\mathcal{N}_i}{P_{i,j}}$. In a lossless power grid, $P_i$ can be formulated using the power flow equations:
\begin{equation} \label{eq:injection}
P_i = \sum\nolimits_{j\in\mathcal{N}_i}V_iV_jB_{ij}\sin \left( \theta _i-\theta _j \right),
\end{equation}
where $V_i$ represents the voltage magnitude of bus $i$ and $B_{ij}$ is the transmission line admittance between bus $i$ and $j$. We consider a linearized and small-signal version of power flow equations \eqref{eq:injection} around the equilibrium point, which can be written as follows:
\begin{equation} \label{eq:flow}
\begin{bmatrix} 
    \boldsymbol{P}_{\mathrm{G}}\\
    \boldsymbol{P}_{\mathrm{L}}
\end{bmatrix}=
\mathbf{H}
\begin{bmatrix} 
    \boldsymbol{\theta }_{\mathrm{G}}\\
    \boldsymbol{\theta }_{\mathrm{L}}
\end{bmatrix} 
=\begin{bmatrix} 
    \mathbf{H}_{\mathrm{GG}}&		\mathbf{H}_{\mathrm{GL}}\\
    \mathbf{H}_{\mathrm{LG}}&		\mathbf{H}_{\mathrm{LL}}
\end{bmatrix}\begin{bmatrix} 
    \boldsymbol{\theta }_{\mathrm{G}}\\
    \boldsymbol{\theta }_{\mathrm{L}}
\end{bmatrix}.
\end{equation}

In this equation, $\mathbf{H}$ represents the power flow Jacobian matrix. The state variables $\boldsymbol{\theta }$ and $\boldsymbol{P}$ are the stacked voltage phasor angles and the network power injections for all buses. The subscripts $\mathrm{G}$ and $\mathrm{L}$ represent the bus corresponding to $\mathcal{N}_{\mathrm{G}}=\{1,...,I_1\}$ and $\mathcal{N}_{\mathrm{L}}=\{I_1+1,...,I\}$, respectively. $\boldsymbol{P}_{\mathrm{G}}=\mathrm{col}(P_i)_{i\in\mathcal{N}_{\mathrm{G}}}, \boldsymbol{P}_{\mathrm{L}}=\mathrm{col}(P_i)_{i\in\mathcal{N}_{\mathrm{L}}}$ and similar notation is used for $\boldsymbol{\theta }$. $\mathbf{H}_{\mathrm{GG}}$ denotes the Jacobian block corresponding to elements $H_{i,j}$ where $i\in\mathcal{N}_{\mathrm{G}}$ and $j\in\mathcal{N}_{\mathrm{G}}$. Similar definitions are adopted for the blocks $\mathbf{H}_{\mathrm{GL}}$, $\mathbf{H}_{\mathrm{LG}}$, and $\mathbf{H}_{\mathrm{LL}}$. The elements $H_{i,j}$ of the power flow Jacobian matrix can be computed as:
\begin{equation}
H_{i,j}=\begin{cases}
    \dfrac{\partial P_i}{\partial \theta _j}=-V_iV_jB_{ij}\cos \theta _{ij},~j\ne i\\
    \dfrac{\partial P_i}{\partial \theta _i}=\sum\nolimits_{j\in\mathcal{N}_i}{V_iV_jB_{ij}\cos \theta _{ij}},~j= i.
\end{cases}
\end{equation}

By defining $\boldsymbol{P}^{\mathrm{ld}}_{\mathrm{L}}=\mathrm{col}(P_i^{\mathrm{ld}})_{i\in\mathcal{N}_{\mathrm{G}}}$, and observing that $\boldsymbol{P}_{\mathrm{L}}=-\boldsymbol{P}^{\mathrm{ld}}_{\mathrm{L}}$ based on Eq.\eqref{eq:load} Thus, the variables $\boldsymbol{\theta }_{\mathrm{L}}$ can be eliminated using equation Eq.\eqref{eq:flow}, resulting in:
\begin{equation} \label{eq:eliminate}
\boldsymbol{P}_{\mathrm{G}}=\left( \mathbf{H}_{\mathrm{GG}}-\mathbf{H}_{\mathrm{GL}}\mathbf{H}_{\mathrm{LL}}^{-1}\mathbf{H}_{\mathrm{LG}} \right) \boldsymbol{\theta }_{\mathrm{G}}-\mathbf{H}_{\mathrm{GL}}\mathbf{H}_{\mathrm{LL}}^{-1}\boldsymbol{P}^{\mathrm{ld}}_{\mathrm{L}},
\end{equation}
where Kron reduction \cite{dorfler2012kron} is used to eliminate all intermediate buses, including load buses. The validity of Kron reduction is based on timescale separation. Loads typically respond much faster than generator units to frequency disturbances, with response times on the order of 0.01 seconds \cite{choi2006measurement}. Integrating Eq.\eqref{eq:eliminate} into Eq.\eqref{eq:converter}, we obtain:
\begin{align}
&\mathbf{M}\dot{\boldsymbol{\omega}}_{\mathrm{G}}+\mathbf{D}\boldsymbol{\omega }_{\mathrm{G}}+\left( \mathbf{H}_{\mathrm{GG}}-\mathbf{H}_{\mathrm{GL}}\mathbf{H}_{\mathrm{LL}}^{-1}\mathbf{H}_{\mathrm{LG}} \right) \boldsymbol{\theta }_{\mathrm{G}}\nonumber
\\
=&\boldsymbol{P}^{\mathrm{mec}}-\boldsymbol{P}_{\mathrm{G}}^{\mathrm{ld}}+\mathbf{H}_{\mathrm{GL}}\mathbf{H}_{\mathrm{LL}}^{-1}\boldsymbol{P}_{\mathrm{L}}^{\mathrm{ld}},
\end{align}
where state variables $\boldsymbol{\omega }$ are the stacked frequencies of all buses and $\boldsymbol{\omega}_{\mathrm{G}}=\mathrm{col}(\omega_i)_{i\in\mathcal{N}_{\mathrm{G}}}$. $\mathbf{M}=\mathrm{diag}\left( m_i \right)_{i\in\mathcal{N}_{\mathrm{G}}},\mathbf{D}=\mathrm{diag}\left( d_i \right)_{i\in\mathcal{N}_{\mathrm{G}}}$ are the diagonal matrices representing the aggregated nodal inertia and damping coefficients for buses ${i\in\mathcal{N}_{\mathrm{G}}}$. 

Defining the reduced Jacobian as $\mathbf{L}:=\mathbf{H}_{\mathrm{GG}}-\mathbf{H}_{\mathrm{GL}}\mathbf{H}_{\mathrm{LL}}^{-1}\mathbf{H}_{\mathrm{LG}}$ and defining the equivalent disturbance as $\boldsymbol{P}^{\mathrm{u}}:=-\boldsymbol{P}_{\mathrm{G}}^{\mathrm{ld}}+\mathbf{H}_{\mathrm{GL}}\mathbf{H}_{\mathrm{LL}}^{-1}\boldsymbol{P}_{\mathrm{L}}^{\mathrm{ld}}$, the state-space model concerning generation unit dynamics is compactly represented as:
\begin{equation} \label{eq:ss}
\begin{bmatrix}
    \boldsymbol{\dot{\theta}}_{\mathrm{G}}\\
    \boldsymbol{\dot{\omega}}_{\mathrm{G}}
\end{bmatrix} =\begin{bmatrix}
    \mathbf{0}&		\mathbf{I}\\
    -\mathbf{M}^{-1}\mathbf{L}&		-\mathbf{M}^{-1}\mathbf{D}
\end{bmatrix} \begin{bmatrix}
    \boldsymbol{\theta }_{\mathrm{G}}\\
    \boldsymbol{\omega }_{\mathrm{G}}
\end{bmatrix} + \begin{bmatrix}
    \mathbf{0}\\
    \mathbf{M}^{-1}
\end{bmatrix} (\boldsymbol{P}^{\mathrm{u}}+\boldsymbol{P}^{\mathrm{mec}}).
\end{equation}

Traditionally, frequency stability and small-signal stability for Eq. \eqref{eq:ss} have been addressed in separate contexts corresponding to different time scales. However, since inertia and damping are key parameters in DER controllers, merely considering their effect on frequency deviations is insufficient. Inappropriate inertia and damping settings can cause instability at the operating point \cite{ma2016research,A.U.-26}, let alone adequately limit frequency deviations under disturbances. Moreover, DERs' responses are significantly faster than those of conventional synchronous machines. Therefore, it is necessary to comprehensively consider the roles of inertia and damping in supporting overall system stability. Additionally, as demonstrated by the well-known structure-preserving model \cite{Hill1981}, positive damping at each node guarantees transient stability; consequently, our primary focus is on small-signal and frequency stability.

\subsection{Oscillation Damping Constraints}
First, we focus on the small-signal stability of the model \eqref{eq:ss}. The inherent mechanical processes within synchronous generators are slow; therefore, we treat mechanical power changes as zero $\boldsymbol{P}^{\mathrm{mec}}=\boldsymbol{0}$ over the time scale of interest. The objective of the oscillation damping constraints is to ensure that the system's dominant oscillation mode, denoted as \(\lambda \in \mathbb{C}\), exhibits sufficient decay speed $\beta$ and maintains a low oscillation damping ratio $\zeta$. The oscillation modes $\lambda$ are the eigenvalues of the state matrix of \eqref{eq:ss}. The eigenvalue $\lambda$ and its eigenvector $[ \boldsymbol{\nu }^{\top}\,\,\boldsymbol{\xi }^{\top} ] ^{\top}$ satisfy the following condition:
\begin{equation} \label{eq:eigenvalue_matrix}
\begin{bmatrix}
    \mathbf{0}&		\mathbf{I}\\
    -\mathbf{M}^{-1}\mathbf{L}&		-\mathbf{M}^{-1}\mathbf{D}
\end{bmatrix}  
\begin{bmatrix} 
    \boldsymbol{\nu }\\
    \boldsymbol{\xi }
\end{bmatrix} =\lambda 
\begin{bmatrix}
    \boldsymbol{\nu }\\
    \boldsymbol{\xi }
\end{bmatrix},
\end{equation}
which can be transformed into scalar quadratic equation sets concerning $\lambda$ via simple algebraic manipulations:
\begin{equation} \label{eq:quadratic_eigenvalue}
\lambda ^2\mathbf{M}\boldsymbol{\nu }+\lambda \mathbf{D}\boldsymbol{\nu }+\mathbf{L}\boldsymbol{\nu }=\boldsymbol{0}.
\end{equation}

\begin{figure}[t]
\centering
\includegraphics[width=0.45\textwidth]{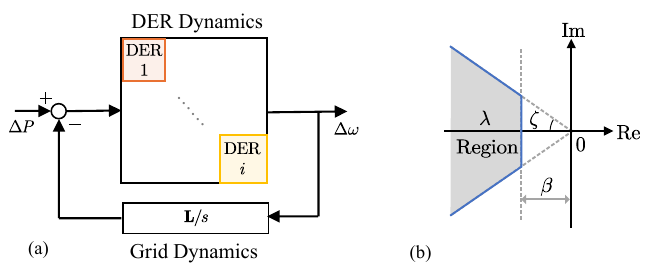}
\caption{(a) The closed-loop control diagram of DER dynamics and power network dynamics. (b) The desired region of system modes.}
\label{fig:block}
\end{figure}
Finding $\lambda$ in the matrix pencil \eqref{eq:quadratic_eigenvalue} is known as solving the quadratic eigenvalue problem (QEP). The QEP is widely used in fields such as dynamic acoustics analysis and fluid mechanics. In these applications, it represents an open-loop arrangement problem that considers modifications to a structure's mass or adjustments to the damping provided by dampers. In contrast, in power systems, inertia and damping are adjustable parameters of DERs, analogous to the PI parameters in a PI controller, and are components of closed-loop controllers. As shown in the closed-loop control block diagram of DERs and the network in Fig. \ref{fig:block}(a), their interactions yield the overall closed-loop response of the system, collectively determining the frequency deviations of the power system.

\color{black}

Notably, \(0\)  is an eigenvalue of the system. This is inherently attributed to the eigenvalue of the Jacobian matrix:
\begin{lemma}[$0$ Eigenvalue~\cite{F.F.-2}] \label{lemma:eig}
$\mathbf{L}$ is positive semi-definite with a $0$ eigenvalue. $0$ is an eigenvalue of Eq.\eqref{eq:ss}'s state matrix.
\end{lemma}
The eigenvalue of \(0\) represents a global shift in the phase angles of all generation units, and angles are not independent states. Angle differences are the true independent variables~\cite{Anderson2003}. The small-signal constraints are about to constrain system modes in the desired region, as Fig. \ref{fig:block}(b).

\textit{Mode Decay}: The location constraint $\mathrm{Re}\left( \lambda \right) \leqslant -\beta$ ensures that transient fluctuations in angles and frequencies do not persist over an extended period. We have the following Proposition to ensure sufficient mode decay: 
\begin{prop} [Mode Decay] \label{th:decay}
The eigenvalues of \eqref{eq:ss} will all be within $\mathrm{Re}\left(\lambda\right) \leqslant -\beta$ except for the eigenvalue $0$ if there exists a positive variable $v$ such that:
\begin{align}\label{eq:plane}
    \mathbf{D}-2\beta \mathbf{M}\succcurlyeq 0,~~\mathbf{L}-\beta \mathbf{D}+\beta ^2\mathbf{M}+v\mathbf{1}\mathbf{1}^{\top}\succcurlyeq 0.
\end{align}
\end{prop}
\begin{proof}
The complete proof is attached in the Appendix. 
\end{proof}

\textit{Mode Damping}: The constraint in the conic region, $\sin \zeta \cdot \mathrm{Re}\left( \lambda \right) +\cos \zeta \cdot \mathrm{Im}\left( \lambda \right) \leqslant 0$, ensures that the system does not undergo serious oscillations. The following proposition provides the condition under which an eigenvalue $\lambda$ lies within the desired conic region:
\begin{prop} [Mode Damping] \label{th:conic}
When $\mathbf{M}\succ \mathbf{0}$, all eigenvalues lie in the conic region $\sin \zeta \cdot \mathrm{Re}\left( \lambda \right) +\cos \zeta \cdot \mathrm{Im}\left( \lambda \right) \leqslant 0$ if and only if the following condition is satisfied:
\begin{equation} \label{eq:conic_original}
    \left< \boldsymbol{\nu },\mathbf{D}\boldsymbol{\nu } \right> ^2-4\left( \cos \zeta \right) ^2\left< \boldsymbol{\nu },\mathbf{L}\boldsymbol{\nu } \right> \left< \boldsymbol{\nu },\mathbf{M}\boldsymbol{\nu } \right> \geqslant 0.
\end{equation}
\end{prop}
\begin{proof}
The complete proof is attached in the Appendix. 
\end{proof}

The constraint in Eq.\eqref{eq:conic_original} is nonconvex because it also depends on the eigenvectors. However, the mode decay constraint Eq.\eqref{eq:plane} provides a natural method for transforming the constraint Eq.\eqref{eq:conic_original} into a convex constraint. By combining Proposition \ref{th:decay} and Proposition \ref{th:conic}, we obtain the complete convex formulations of constraints for small-signal stability:
\begin{theorem} [Mode Placement Constraints] \label{th:placement}
If the following condition is satisfied, all eigenvalues are in the left half-plane $\lambda \leqslant -\beta $ and the conic region $\sin \zeta \cdot \mathrm{Re}\left( \lambda \right) +\cos \zeta \cdot \mathrm{Im}\left( \lambda \right) \leqslant 0$ simultaneously:
\begin{align} \label{eq:eigenvalue_constraints}
    \text{Constraint}~\eqref{eq:plane},~~\beta\mathbf{D}-2\left( \cos \zeta \right) ^2\mathbf{L}\succcurlyeq 0.
\end{align}
\end{theorem}

The small-signal stability constraints in Eq. \eqref{eq:eigenvalue_constraints} are all expressed as positive semi-definite constraints. These nonlinear but convex matrix inequalities partially capture the complex interactions between the network matrix $\mathbf{L}$ and the DERs. 

Notice Theorem \ref{th:placement} presents sufficient parametric conditions for constraining the eigenvalue locations. Although necessary and sufficient conditions for eigenvalue placement can be derived using linear matrix inequalities (LMIs), these conditions result in non-convex bilinear matrix inequalities. The proposed convex constraints are computationally efficient but can also be conservative. Nevertheless, appropriate over-estimations can enhance the robustness of the allocation plan.

Besides, Theorem \ref{th:placement} does not require the uniform-ratio assumption. Only when constraint Eq. \eqref{eq:plane} is strictly binding (\(\mathbf{D}=2\beta \mathbf{M}\)) for all nodes, all DERs will adopt the uniform inertia/damping ratio \(2\beta\). In this case, the method simplifies to the oscillation mode calculation approaches indicated by extensive research~\cite{SajadiKenyon-30,M.X.-29,GuoZhao-28}.

\subsection{Frequency Stability Constraints}
Here, we focus on the frequency stability of the model \eqref{eq:ss}. Frequency stability is usually concerned with the system's response over a longer time period following a significant step disturbance. Slow dynamics, such as generator turbine dynamics, and steady-state frequency deviations are typically considered. On the time scales of interest, the mechanical dynamics of the bus turbine $i$ is:
\begin{equation} \label{eq:turbine}
\tau_i \dot{P}^{\mathrm{mec}}_{i}=-P^{\mathrm{mec}}_{i}-r_i^{-1}\omega_{i},
\end{equation}
where $r_i$ is the turbine droop coefficient and $\tau_i$ is the turbine time constant. Eq.\eqref{eq:turbine} indicates that the mechanical part of synchronous generators injects power according to frequency deviations. But the response is not instantaneous and has a power ramping delay $\tau_i$. Keeping in line with traditional methodology, we analyze frequency stability in a global coordinate, usually denoted by the center of inertia (COI) frequency. The differential equations governing the COI frequency dynamics are~\cite{Global2020,guggilam2018optimizing}:
\begin{align}
&m\dot{\omega}_{\mathrm{C}}+d\omega _{\mathrm{C}}=P^{\mathrm{u}}+P^{\mathrm{mec}}, \label{eq:turbine1}
\\
&\tau \dot{P}^{\mathrm{mec}}=-P^{\mathrm{mec}}-r^{-1}\omega_{\mathrm{C}}, \label{eq:turbine2}
\end{align}
where $\omega _{\mathrm{C}}$ is the COI frequency and $P^{\mathrm{mec}}$ is aggregated turbine dynamics. In Eq.\eqref{eq:turbine1}, $m=\sum_i m_i$, \(d=\sum_i d_i\) and $P^{\mathrm{u}}=\sum_i P^{\mathrm{u}}_{i}$ represent the total inertia, damping, and disturbance of the system, respectively. In Eq.\eqref{eq:turbine2}, $\tau=(\sum_i \tau_i)^{-1}$ and $r=\sum_i r_i$ represent the equivalent turbine constant and droop coefficient for synchronous generators~\cite{Global2020,guggilam2018optimizing}, respectively. Since the aggregated generator turbine dynamics of \(r\) and \(\tau\) are independent of \(m\) and \(d\), they can be treated as given.

According to Eqs.\eqref{eq:turbine1}-\eqref{eq:turbine2}, $\omega _{\mathrm{C}}\left( s \right)$ following a step disturbance can be represented as:
\begin{equation} \label{eq:omega_coi}
\omega _{\mathrm{C}}\left( s \right) =\frac{P^{\mathrm{u}} \left( \tau s+1 \right)}{s\left( m\tau s^2+\left( m+d\tau \right) s+d+r^{-1} \right)}.
\end{equation}

Common indicators of frequency stability performance in the system include the rate of change of frequency (RoCoF), frequency nadir, and the steady-state frequency deviations after disturbance. The maximum RoCoF (obtained at the time instant $t=0^+$ \cite{Global2020}) and the steady-state frequency can be computed as: 
\begin{align}
& \left| \dot{\omega}_{\mathrm{C}}(t) \right|_{\infty}=\underset{s\rightarrow \infty}{\lim}s\dot{\omega}_{\mathrm{C}}\left( s \right) =m^{-1}\left| P^{\mathrm{u}}  \right|,\label{eq:initial}
\\
& \left| \omega _{\mathrm{C}}\left( \infty \right) \right|=\underset{s\rightarrow 0}{\lim}s \omega _{\mathrm{C}}\left( s \right) =\left( d+r^{-1} \right) ^{-1}\left| P^{\mathrm{u}}  \right|.\label{eq:synchronized}
\end{align}

The RoCoF should be below $\varepsilon _{\mathrm{RoCoF}}$. The maximum allowed steady-state frequency deviation is $\varepsilon _{\mathrm{sync}}$. They are constrained as follows:
\begin{equation}\label{eq:COI_rocof}
\left| P^{\mathrm{u}} \right|\leqslant \varepsilon _{\mathrm{RoCoF}}m_{\Sigma},~~\left| P^{\mathrm{u}}  \right|\leqslant \varepsilon _{\mathrm{sync}}\left( d+r^{-1} \right).
\end{equation}

The COI RoCoF constraint serves as an approximation for individual buses. In industrial practice, the RoCoF is measured over a period of time for each bus, rather than as the instantaneous derivative. Therefore, it is crucial to incorporate both the COI RoCoF constraint and oscillation damping constraints to ensure that the real RoCoF values are satisfactory.

By setting \(\dot{\omega} _{\mathrm{COI}}\left( t \right) = 0\), we can determine the time of the frequency nadir and subsequently derive the expressions for the frequency nadir.
\begin{lemma}[Frequency Nadir \cite{Global2020}]
Given the power system COI frequency dynamics via Eq.\eqref{eq:omega_coi}, the COI frequency nadir is given by:
\begin{equation} \label{eq:nadir_time}
    \left| \omega _{\mathrm{COI}}(t)  \right|_{\infty}=\frac{\left| P^{\mathrm{u}}  \right|}{d+r^{-1}}\left( 1+\sqrt{\frac{\tau r^{-1}}{m}}e^{-\frac{\eta}{\omega _{\mathrm{d}}}\left( \phi +\frac{\pi}{2} \right)} \right),
\end{equation}
where $\omega _{\mathrm{d}}$, $\eta$, and $\phi$ are uniquely determined by:
$$
\begin{aligned}
    &\omega _{\mathrm{d}}=\sqrt{\frac{d+r^{-1}}{m\tau}-\frac{1}{4}\left( \frac{1}{\tau}+\frac{d}{m} \right) ^2}, \eta =\frac{1}{2}\left( \tau ^{-1}+\frac{d}{m} \right),
    \\
    &\sin(\phi )=\frac{\left( \tau ^{-1}-\eta \right)}{\sqrt{\omega _{\mathrm{d}}^{2}+\left( \tau ^{-1}-\eta \right) ^2}}=\frac{m-d\tau}{2\sqrt{m\tau r^{-1}}}.
\end{aligned}
$$
For analysis, upon knowing the value of $\tau$ and $r$, we denote the inverse of the frequency nadir value as a function of $m$ and $d$ as $g \left( m,d \right) = \left| \omega _{\mathrm{COI}}(t)  \right|_{\infty}^{-1}$.
\end{lemma}
Direct attempts to constrain the frequency nadir using Eq.\eqref{eq:nadir_time} result in high non-convexity. Instead, we resort to approximating the nonlinear function $g \left( m,d \right)$ as its first-order Taylor expansion at the point $(m_{ \left( 0 \right)},d_{ \left( 0 \right)})$:
\begin{align}
\hat{g}\left( m,d \right) \approx& g \left( m_{ \left( 0 \right)},d_{ \left( 0 \right)} \right) +\left. \nabla _{m}g \right|_{m_{ \left( 0 \right)},d_{ \left( 0 \right)}}\left( m-m_{ \left( 0 \right)} \right)
\nonumber
\\
& +\left. \nabla _{d}g \right|_{m_{ \left( 0 \right)},d_{ \left( 0 \right)}}\left( d-d_{ \left( 0 \right)} \right), \label{eq:emprical_function}
\end{align}
where the approximation error is proportional to $|m-m_{ \left( 0 \right)}|+|d-d_{ \left( 0 \right)}|$ according to Peano's form of the remainder. The reason for the goodness of this linear approximation can be explained via numerical analysis concerning the \textit{Gaussian curvature} $\kappa$ of $g \left( m,d \right)$, which is computed as:
\begin{equation} \label{eq:curvature}
\kappa =\frac{\nabla _{mm}g\cdot \nabla _{dd}g-\nabla _{md}g}{1+\left( \nabla _{m}g \right) ^2+\left( \nabla _{d}g \right) ^2}.
\end{equation}

It can be verified that $|\kappa|$ is generally below $10^{-3}$ almost everywhere when $\tau$ and $r$ vary. It means that $\hat{g}\left( m,d \right)$ is not seriously `distorted' and can be well approximated by a linear plane in \eqref{eq:emprical_function}. It is recommended that $(m_{ \left( 0 \right)},d_{ \left( 0 \right)})$ can be selected within the normal operational ranges, such that the linear approximation error can be further reduced, typically with error $\leqslant 3\%$. Given a threshold for allowed frequency nadir, \(\varepsilon _{\mathrm{nadir}}\), we can express the approximated constraint as:
\begin{equation} \label{eq:nadir_constraint}
\hat{g}\left( m,d \right) \geqslant \varepsilon _{\mathrm{nadir}}^{-1},
\end{equation}
where $\hat{g}\left( m,d \right)$ is the linear approximation given in \eqref{eq:emprical_function}.

\section{Optimization Problem Formulation and Incentive Design}\label{VCG}
\subsection{Cost Function}
The changes in output power $P_{\mathrm{o},i}$ and energy $E_{\mathrm{o},i}$ for generation unit $i$ are as follows:
\begin{align}\label{eq:output}
&P_{\mathrm{o},i} = m_i\dot{\omega}_i+d_i\omega _i,~~E_{\mathrm{o},i}=m_i\omega _i+d_i\int_0^t{\omega _i\left( \tau \right) \,\,\mathrm{d}\tau}.
\end{align}
DERs need to deviate from their maximum power tracking points to reserve the necessary power. This cost, known as the \textit{power reserve cost} $C_{i,\mathrm{p}}$, is related to the maximum reserve power, denoted by $\|P_{\mathrm{o},i}(t)\|_{\infty}$. Additionally, batteries, serving as primary sources for some DERs, must discharge a fraction of their stored energy to support the grid. This cost component, associated with the maximum energy loss $\|E_{\mathrm{o},i}(t)\|_{\infty}$, is \textit{energy reserve cost} $C_{i,\mathrm{e}}$.

Different components of the output power and energy, as described in \eqref{eq:output}, incur variable costs across various types of generation units. The term \(  m_i\dot{\omega}_i(t) \) denotes the inertia buffer power. Wind turbines with rotating mechanics can provide this inertial value \(m_i\) at a lower cost. Other static DERs must reserve the maximum power to deliver \( m_i\dot{\omega}_i(t) \). The energy $m_i\omega_i$ can be charged back automatically after the secondary frequency control returns the system frequency to the nominal value. For the other part $ d_i\int {\omega _i(\tau )\mathrm{d}\tau} $, no existing controllers will increase the frequency above the nominal value to recover this part of energy. Consequently, the costs of output power and energy vary. Taking the power reserve cost, $C_{i,\mathrm{p}}$, as an example, it consists of two distinct components:
\begin{equation}
    \begin{aligned}
   C_{i,\mathrm{p}}=\alpha _{i,\mathrm{p}}\left( m_i\left\| \dot{\omega}_i(t) \right\| _{\infty} \right) ^2+\beta _{i,\mathrm{p}}m_i\left\| \dot{\omega}_i(t) \right\| _{\infty}
    \\
    \,\,   +\xi _{i,\mathrm{p}}\left( d_i\left\| \omega _i(t) \right\| _{\infty} \right) ^2+\zeta _{i,\mathrm{p}} d_i\left\| \omega _i(t) \right\| _{\infty},     
    \end{aligned}
\end{equation}
where $\alpha _{i,\mathrm{p}}$ and $\beta _{i,\mathrm{p}}$ are cost coefficients for the reserve power $m_i\dot{\omega}_i(t)$, and $\xi _{i,\mathrm{p}}$ and $\zeta _{i,\mathrm{p}}$ for $d_i\omega_i(t)$. Quadratic costs are commonly used to reflect that costs may increase with the provision of deloading reserves. For an individual DER, determining its impact on system dynamics can be challenging. System operators can provide DERs with anticipated worst-case values of \(\left\| \omega_i(t) \right\|_{\infty}\) and \(\left\| \dot{\omega}_i(t) \right\|_{\infty}\) in advance. These values are then treated as known parameters by DERs. For instance, one can estimate reserve power by multiplying the system operator’s RoCoF constraint value by the DER’s inertia $m$ or by taking the product of the system operator’s steady-state frequency deviation constraint value and the DER’s damping $d$. We illustrate the estimation methods in detail in the case studies.

Thus, $C_{i,\mathrm{p}}$ becomes a quadratic function of $m_i$ and $d_i$. This approach also applies to the energy reserve cost, $C_{i,\mathrm{e}}$. Combining these considerations, the total cost of the DER $i$ can be expressed as:
\begin{equation}
\begin{aligned}
\mathcal{C} _i(m_i,d_i)&=C_{i,\mathrm{p}} + C_{i,\mathrm{e}}
\\
&=\varrho _{\mathrm{m},i}m_{i}^{2}+\mu _{\mathrm{m},i}m_i+\varrho _{\mathrm{d},i}d_{i}^{2}+\mu _{\mathrm{d},i}d_i,
\end{aligned}
\end{equation}
where $\varrho _{\mathrm{m},i}$ and $\mu _{\mathrm{m},i}$ are equivalent cost coefficients for inertia, and $\varrho _{\mathrm{d},i}$ and $\mu _{\mathrm{d},i}$ for damping.

\subsection{Complete Optimization Problem}
In summary, the complete cost minimization problem of oscillation damping and inertia is formulated as follows:
\begin{equation}\label{eq:final_model}
\begin{aligned}
    &\min~ \sum\nolimits_{i}\mathcal{C} _i(m_i,d_i)
    \\
    &~\mathrm{s.t.}~~\textrm{Small-Signal~Stability~Constraints:~}\eqref{eq:eigenvalue_constraints}
    \\  &~~~~~~~\textrm{Frequency~Stability~Constraints:~}\eqref{eq:COI_rocof},\eqref{eq:nadir_constraint}
    \\
    &~~~~~~~m_i~\textrm{and}~d_i~\textrm{Range~Constraints}.
\end{aligned}	
\end{equation}
The decision variables are all DERs' inertial and damping coefficients, as well as auxiliary variables $v$ in \eqref{eq:plane}. The ranges of $m_i$ and $d_i$ are determined according to the primary source of DERs. Problem \eqref{eq:final_model} is a convex positive semi-definite programming model. Recent grid codes, such as those by National Grid and AEMO, establish guidelines for selecting constants in grid-forming converters, including criteria for test validations of these requirements. This approach can be employed to assist in the network-wide planning for controller parameter tuning in a cost-effective manner.

\subsection{Realistic Considerations}
The standard model in Eq. \eqref{eq:final_model} assumes perfect knowledge of the system parameters and load conditions, which cannot consider varying grid conditions. To address this limitation, we enhance robustness by explicitly incorporating uncertainties in grid conditions, network topology, and system parameters while preserving convexity.

First, load conditions and network parameters are reflected in the power flow Jacobian matrix $\mathbf{L}$. The perturbed power flow Jacobian, denoted as \( \mathbf{L} + \Delta \mathbf{L} \), captures potential deviations due to changes in network topology, parameter estimation errors, or variations in grid conditions, where \( \Delta \mathbf{L} \) belongs to the uncertainty set \( \mathcal{U}_{\mathrm{g}} \). For load variations, $\mathcal{U}_{\mathrm{g}} $ can be defined as the union set of typical load conditions, namely \( \mathcal{U}_{\mathrm{g}}:=  \{\Delta \mathbf{L}|\Delta \mathbf{L}=\mathbf{L}_{m}-\mathbf{L}\}\), where $\mathbf{L}_{m}$ represents the power flow Jacobian matrices corresponding to all predefined load conditions considered by the system operator. For network uncertainty, line parameters can vary by $\eta$\% of their standard values, yielding the uncertainty set \( \mathcal{U} =  [-\eta,\eta]\%\sum_{i,j,i<j}\mathbf{L}_{\mathrm{b}(i,j)}\). Here,
$$
\mathbf{L}_{\mathrm{b}(i,j)} = l_{i,j}\left[
\begin{array}{ccccc}
	\cdots&		\cdots&		\cdots&		\cdots&		\cdots\\
	\cdots&		\displaystyle\overset{i,i}{-1} &	\cdots&		\displaystyle\overset{i,j}{1} &	\cdots\\
	\cdots&		\cdots&		\cdots&		\cdots&		\cdots\\
	\cdots&		\displaystyle\overset{j,i}{1} &	\cdots&		\displaystyle\overset{j,j}{-1} &	\cdots\\
	\cdots&		\cdots&		\cdots&		\cdots&		\cdots\\
\end{array}\right],
$$
where $l_{i,j}$ is the $(i,j)$-th element of $\mathbf{L}$ and $\cdots$ indicates the corresponding  elements equal to 0. This structured perturbation ensures that \( \mathbf{L} + \Delta \mathbf{L} \) remains a valid power flow Jacobian matrix. The robust small-signal stability constraints then become: 
\begin{align}
&\mathbf{D}-2\beta \mathbf{M}\succcurlyeq 0,
\\
&\mathbf{L}+\Delta \mathbf{L}-\beta \mathbf{D}+\beta ^2\mathbf{M}+v\mathbf{1}\mathbf{1}^{\top}\succcurlyeq 0, \forall \Delta \mathbf{L} \in \mathcal{U}_{\mathrm{g}}, \label{eq:robust1}
	\\
&\beta\mathbf{D}-2\left( \cos \zeta \right) ^2(\mathbf{L}+\Delta \mathbf{L})\succcurlyeq 0, \forall \Delta \mathbf{L} \in \mathcal{U}_{\mathrm{g}}.\label{eq:robust2}
\end{align}

Similarly, inertia and damping uncertainties, denoted by $\Delta \mathbf{M} = \mathrm{diag}(\Delta m_i)$ and $\Delta \mathbf{D} = \mathrm{diag}(\Delta d_i)$, respectively, are introduced to account for variability in synchronous machines and loads. The uncertainty sets of inertia and damping are denoted by $\mathcal{U}_{\mathrm{m}}$ and $\mathcal{U}_{\mathrm{d}}$, respectively. 

In summary, the complete robust cost minimization problem is formulated as:
\begin{align}
    \min~ &\sum\nolimits_{i}\mathcal{C} _i(m_i,d_i) \nonumber
    \\
    ~\mathrm{s.t.}~&\textrm{Robust~Small-Signal~Stability~Constraints:~} \nonumber\\
    &~~\mathbf{D}+\Delta \mathbf{D}-2\beta(\mathbf{M}+\Delta \mathbf{M})\succcurlyeq 0, \nonumber\\
    &~~\mathbf{L}+\Delta \mathbf{L}-\beta (\mathbf{D}+\Delta \mathbf{D})+\beta ^2 (\mathbf{M}+\Delta \mathbf{M})+v\mathbf{1}\mathbf{1}^{\top}\succcurlyeq 0, \nonumber\\ 
    &~~\beta(\mathbf{D}+\Delta \mathbf{D})-2\left( \cos \zeta \right) ^2(\mathbf{L}+\Delta \mathbf{L})\succcurlyeq 0, \nonumber\\
    & ~~~~~~~~~~~~\forall~\Delta \mathbf{L} \in \mathcal{U}_{\mathrm{g}}, \Delta \mathbf{M}\in \mathcal{U}_{\mathrm{m}}, \Delta \mathbf{D}\in \mathcal{U}_{\mathrm{d}}. \nonumber\\
    &\textrm{Robust~Frequency~Stability~Constraints:~} \label{eq:robust_model}\\
    &~~\hat{g}\left( m+\Delta m,d+\Delta d \right) \geqslant \varepsilon _{\mathrm{nadir}}^{-1},\\
    & ~~~~~~~~~~~~\forall~\Delta \mathbf{M}\in \mathcal{U}_{\mathrm{m}}, \Delta \mathbf{D}\in \mathcal{U}_{\mathrm{d}}. \nonumber\\		
    &m_i~\textrm{and}~d_i~\textrm{Range~Constraints} \nonumber.
\end{align}	
 
The resulting model retains convexity as a positive semi-definite program, enabling computationally efficient solutions via robust optimization techniques.

Due to technical constraints or other considerations, some DERs may be reluctant to provide inertia and damping. Moreover, system operators may wish to minimize the number of DERs that supply these parameters. In such cases, a sparsity-promoting penalty term can be incorporated into the objective function. Ideally, the \(\mathcal{L}_0\) norm of inertia and damping would be added to account for the number of nonzero parameters. However, because the \(\mathcal{L}_0\) norm is non-convex, a more common practice is to include an \(\mathcal{L}_1\) norm penalty term in the objective function \cite{candes2008enhancing} to encourage the associated inertia or damping values to be zero.

\subsection{Incentive Design}

The increasing penetration of renewable energy has prompted many countries to reform grid regulations and introduce novel control and stabilization services that use economic incentives to encourage various power sources to support system stability. Moreover, the range of ancillary services is expected to continue expanding. For instance, in Finland,where renewable energy comprises 40\% of the energy mix, several new ancillary service markets have been introduced in recent years to maintain system stability, thereby providing a solid foundation for secure grid operation \cite{lieskoski2024review}. In this subsection, we propose an economic incentive mechanism that has the potential to be employed for ancillary services. This model and approach aim to leverage economic incentives to maintain system stability in the future.

The complete incentive design for the joint oscillation damping and inertia provision service is as follows, which can be seen as a new ancillary service market mechanism:
\begin{itemize}[leftmargin=*]
\item \textit{Bidding}: Each DER submits its bids for $m_i$ and $d_i$ as $\mathcal{C}_i\left( m_{i},d_{i} \right)$ to provide $m_i$ and $d_i$.
\item \textit{Clearing}: The system operator collects all costs and determines the amount of $m_i^{\star},$ and $d_i^{\star}$ purchased from each converter by solving the optimization problem \eqref{eq:final_model}.
\item \textit{Payment}: The system operator pays $q_i$ to unit $i$ for providing $m_i$ and $d_i$.
\end{itemize}
In the bidding step, a DER submits parameters $\varrho _{\mathrm{m},i},\mu _{\mathrm{m},i},\varrho _{\mathrm{d},i}$ and $\mu _{\mathrm{d},i}$. In real implementations, the converters can submit a block-shaped, monotonically increasing bid just like bids in the energy market. Upon completion, it is essential for the system operator to validate that the converters are delivering inertia and damping coefficients in accordance with the specified allocations. Failure to meet these specifications will result in nonpayment. Consequently, all DERs can participate in the bidding process that reflects their equivalent damping capabilities, including those provided via PSSs. However, it is important to note that their bids should be limited to their deterministic damping capabilities.

In this paper, the VCG mechanism \cite{xu2015efficient} is employed to determine the payment rules. The payment to unit $i$ is defined as follows:
\begin{equation}
\begin{aligned}
    q_i=\underset{m_j,d_j}{\min}\sum\nolimits_{j\ne i}^{}{\mathcal{C} _j\left( m_j,d_j \right)}-\sum\nolimits_{j\ne i}^{}{\mathcal{C} _j\left( m_{j}^{\star},d_{j}^{\star} \right)}
\end{aligned}
\end{equation}
The first part ${\min\nolimits_{m_j,d_j}}\sum\nolimits_{j\ne i}^{}{\mathcal{C} _j\left( m_j,d_j \right)}$ represents the cost to maintain the stability of the system without the presence of unit $i$. It takes into account the aggregated cost of all other DERs in the system. The second part of the equation represents the summed cost for all other converters with the presence of DER $i$. The payment is determined by the difference between the two parts, representing the reduction in the overall cost achieved by having DER $i$ participate in the service. Namely, each DER gets paid according to the negative externality. For rational generation, the VCG mechanism is proven to satisfy \textit{Cost Minimization}, \textit{Individual Rationality}, and \textit{Incentive Compatible} property~\cite{xu2015efficient}, which can guarantee that the DERs submit their true costs and the cost of providing inertia and damping is minimized.

The effectiveness of the framework above depends on the presence of market-driven incentives or regulatory changes. Meanwhile, the proposed stability constraints are not limited to market applications but are also applicable to various operational optimization and planning challenges, such as determining the construction strategy for future inverter-based power plants by appropriately allocating damping and inertia at each node. This approach can reduce costs while ensuring the stability of power systems with high renewable energy penetration. Moreover, because the constraints are convex with respect to the system topology-related matrix $\mathbf{L}$, system operators can optimize system networks in planning problems using our method. Additionally, it is possible to plan for the required inertia and damping of synchronous machines to meet system stability requirements, rendering synchronous generators' damping and inertia adjustable factors in the planning process.

\section{Case Studies}\label{case}
\begin{figure}[t]
\centering
\includegraphics[width=0.5\textwidth]{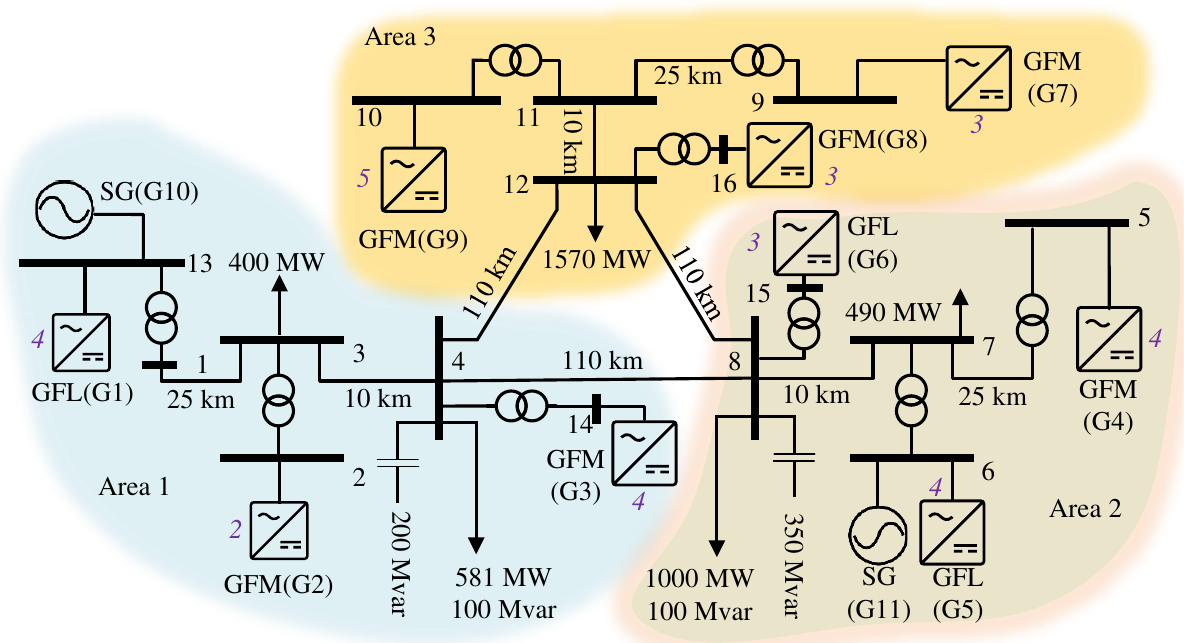}
\caption{System topology for the basic test setup. Each node's number of units is indicated by numbers in Italics adjacent to it.}
\label{fig:case}
\end{figure}
\begin{figure}[t]
\centering
\includegraphics[width=0.35\textwidth]{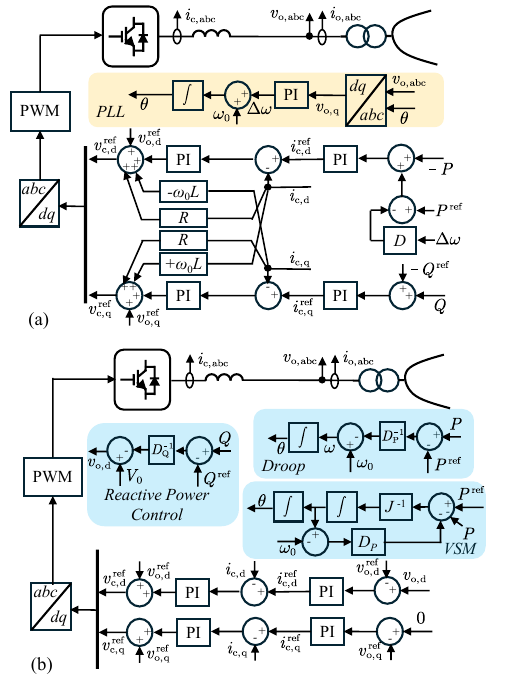}
\caption{(a) The GFL DER controller structure. (b) The GFM DER controller structure.}
\label{fig:loop}
\end{figure}

This study presents a case analysis of a three-area, twelve-bus test system shown in Fig. \ref{fig:case}(a), which directly extends the classical Kundur two-area system widely used in inter-area oscillation studies. An electromagnetic transient Simulink model was employed for the test system. The setup comprises 6 GFM DERs, 3 GFL DERs, and 2 synchronous generators (SGs), with a base capacity of 100 MW for the entire system and for each unit. The synchronous machines are equipped with an IEEE-PSS1A power system stabilizer (PSS) and an ST1A automatic voltage regulator (AVR). Their governors are modeled as droop-controlled turbines with a first-order delay; turbine parameters are set to \( r^{-1}_i=10 \) MW$\cdot\mathrm{s}$/rad and \( \tau_i=5 \) s. Each SG is characterized by a fixed rotational inertia of $1$ s. The damping coefficient of SGs is set as 0. 

For GFL DERs, a synchronous reference frame PLL-based control architecture with an active power droop controller is adopted. For GFM DERs, a VSM synchronization loop is utilized if the corresponding converter is allocated inertia; otherwise, droop control is employed. Detailed configurations of the LR filters and inner-loop controllers are provided in Fig. \ref{fig:loop}(a-b). For brevity, the full model is open-sourced, and all controller and filter parameters are available in the public repository. Additionally, all controller, filter, and system operating point parameters can be found in the online files. \footnote{\href{https://github.com/VictorCFeng/Hybrid-Oscillation-Damping-and-Inertia-Management-for-DERs.git}{https://github.com/VictorCFeng/Hybrid-Oscillation-Damping-and-Inertia-Management-for-DERs.git}}

The oscillation mode requirement is set to $\beta=3$ and $\cos \zeta = 0.1$, according to the grid code's requirement in \cite{Requirements}. The COI RoCoF limit, steady-state frequency deviation, and frequency nadir limits are set as $\varepsilon _{\mathrm{RoCoF}}=1$ Hz, $\varepsilon _{\mathrm{sync}}=0.1$ Hz, $\varepsilon _{\mathrm{nadir}}=0.3$ Hz, respectively. Individual inertia and damping limits are set at $100$ MW$\cdot\mathrm{s}^2$/rad and $100$ MW$\cdot\mathrm{s}$/rad, respectively. The system is subjected to a step load increase of $300$ MW.  The expected electricity price $\mathfrak{p}$ varies from $20$ to $50$ \$/MWh~\cite{price}. We assume multiple converters are located at each node, and the number of converters is shown in Fig.\ref{fig:case}. Each DER is owned by a different manufacturer. For each DER $k$, $\mu _{\mathrm{m},k}=\mathfrak{p}_k\varepsilon _{\mathrm{RoCoF}}$ and $\mu _{\mathrm{d},k}=\mathfrak{p}_k\varepsilon _{\mathrm{nadir}}$ are set, where $\mathfrak{p}_k$ is randomly selected within $[20, 50]$. The quadratic coefficients are set as $\varrho _{\mathrm{m},k}=\mu _{\mathrm{m},k}/50$ and $\varrho _{\mathrm{d},k}=\mu _{\mathrm{d},k}/50$.  If two DERs at the same node submit identical winning bids, they will split the winning bids and receive an equal distribution of the inertia and damping allocation.

The linearized nadir expression \eqref{eq:emprical_function} is expanded at the point $(m_{(0)},d_{(0)})=(200 \mathrm{~MW\cdot s^2/rad},200 \mathrm{~MW\cdot s/rad})$. Numerical results indicate that using $\hat{g} \left( m ,d \right)$ to approximate $g\left( m ,d \right)$ results in a mean error of $1.88\%$ within the range $\left( m,d \right) \in \left[ 10,10^3 \right] \mathrm{~MW\cdot s^2/rad} \times \left[ 10,10^3 \right] \mathrm{~MW\cdot s/rad}$. All experiments are implemented using Matlab and Simulink on a workstation with a 3.70 GHz CPU and 128.0 GB of RAM. The optimization problem is solved using the Mosek solver.

\subsection{Allocation Result}

We first analyze the nodal inertia and damping coefficient allocation results. For comparison, we consider three scenarios:

\begin{itemize}[leftmargin=*]
\item Scenario 1: The proposed standard model incorporating both small-signal and frequency stability constraints is solved using the commercial convex optimization solver Mosek.
\item Scenario 2: The proposed optimization model considers only frequency stability constraints, with these constraints expressed via a Taylor approximation, and is solved using the commercial convex optimization solver Mosek.
\item Scenario 3: The proposed optimization model considers only frequency stability constraints, but the stability constraints are formulated non-convexly. Due to the non-convexity, the genetic algorithm is employed to solve it.
\end{itemize}

Scenario 2 represents the typical approach in current frequency stability and inertia market studies \cite{SajadiKenyon-30,M.X.-29,L.F.-23,Z.R.-24,B.D.-3}, in which only the effects of inertia and damping on frequency deviation and RoCoF are considered. Scenario 3 reflects existing studies employing nonlinear expressions to model the nadir \cite{Global2020,guggilam2018optimizing}, thereby serving as benchmarks for assessing the impact of non-convexity on the solution process.

\begin{figure}[t]
\centering
\includegraphics[width=0.5\textwidth]{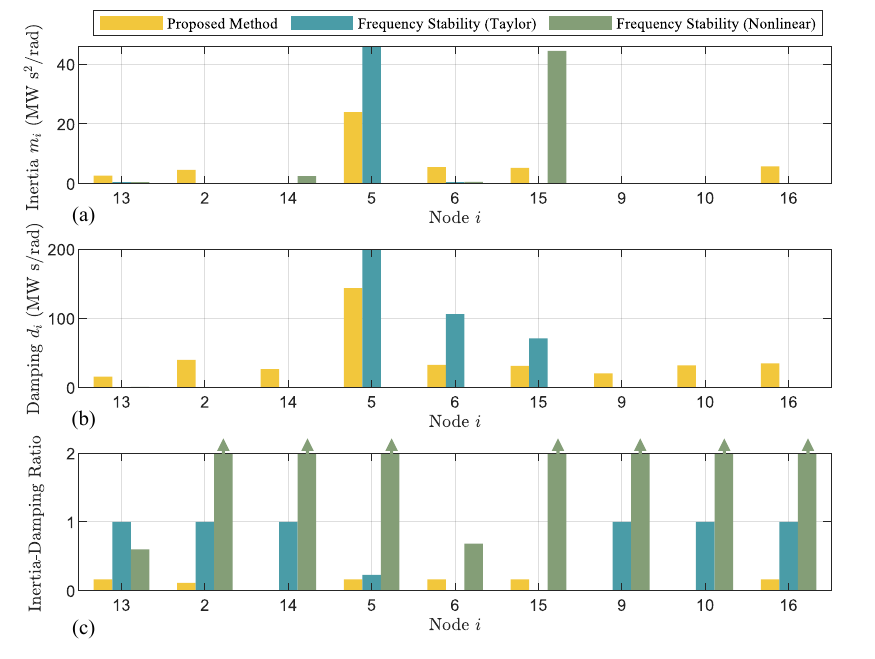}
\caption{(a) The allocation of $m_i$, (b) the allocation of $d_i$ and (c) the inertia-damping ratio of different nodes in 3 different scenarios. In (c), the upper arrow means the ratio is larger than 2.}
\label{fig:placement}
\end{figure}

\begin{figure}[t]
\centering
\includegraphics[width=0.5\textwidth]{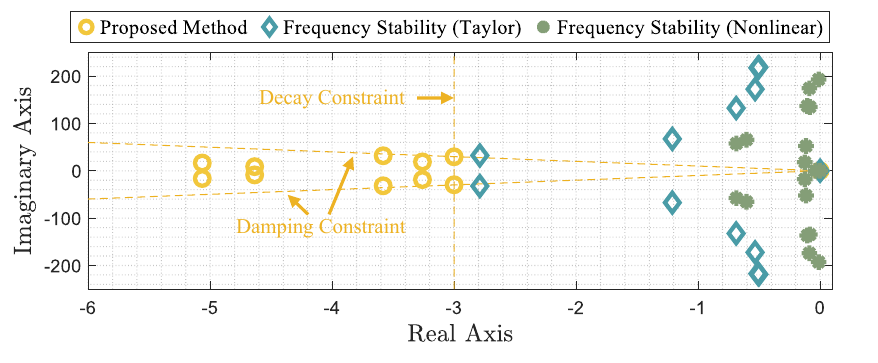}
\caption{The eigenvalue maps of 3 different scenarios. The sufficient constraints presented in Theorem 1 are tight, since all complex eigenvalues are positioned at the boundary.}
\label{fig:eig}
\end{figure}

The nodal distributions of inertia and damping, as well as the inertia-damping ratios at each node, are shown in Fig. \ref{fig:placement}(a-c). In both Scenarios 1 and 2, the total inertia is \(\sum_i m_i = 47.75\) MW·s$^2$/rad and the total damping is \(\sum_i d_i = 378.53\) MW·s/rad. In contrast, Scenario 3 yields a total inertia of \(\sum_i m_i = 48.79\) MW·s$^2$/rad but a total damping of only \(\sum_i d_i = 1.83\) MW·s/rad. This near-zero total damping in Scenario 3 arises from the algorithm’s difficulty in handling non-convex constraints involving inverse trigonometric and exponential functions, resulting in a low-quality solution. These findings underscore the importance of convex constraints in achieving both efficient and reliable solutions. The eigenvalue maps for the three scenarios are depicted in Fig. \ref{fig:eig}. Dashed lines represent the boundaries of the desired eigenvalue region defined in the optimization problem ($\lambda\leqslant\beta$ and $\sin \zeta \cdot \mathrm{Re}\left( \lambda \right) +\cos \zeta \cdot \mathrm{Im}\left( \lambda \right) \leqslant 0$). In Fig. \ref{fig:eig}, Scenario 2 and Scenario 3 both yield pairs of eigenvalues with very poor damping performance. In contrast, Scenario 1's solutions effectively constrain the eigenvalues within the desired region.

Furthermore, it is observed that the proposed convex small-signal stability constraints do not increase the total damping value; rather, they effectively redistribute the damping to satisfy the additional stability constraints while still meeting the frequency stability requirements. This indicates that although positive semi-definite constraints may introduce some conservatism, the overestimation of inertia and damping is minimal. As shown in Fig. \ref{fig:placement}(c), Scenario 1 does not enforce all nodes to exhibit a uniform inertia-damping ratio. For instance, some nodes (Node 14, 9, and 10) are not allocated any inertia. This observation confirms the conclusion from our theoretical derivation that a uniform inertia-damping ratio across all nodes is not required. The nearly uniform inertia-damping ratio observed in Scenario 2 arises because the prices for inertia and damping are set proportionally. Consequently, a node that is economical in providing inertia will also be economical in providing damping. Altering the price settings can lead to non-uniform inertia-damping ratios.

\begin{figure}[t]
\centering
\includegraphics[width=0.5\textwidth]{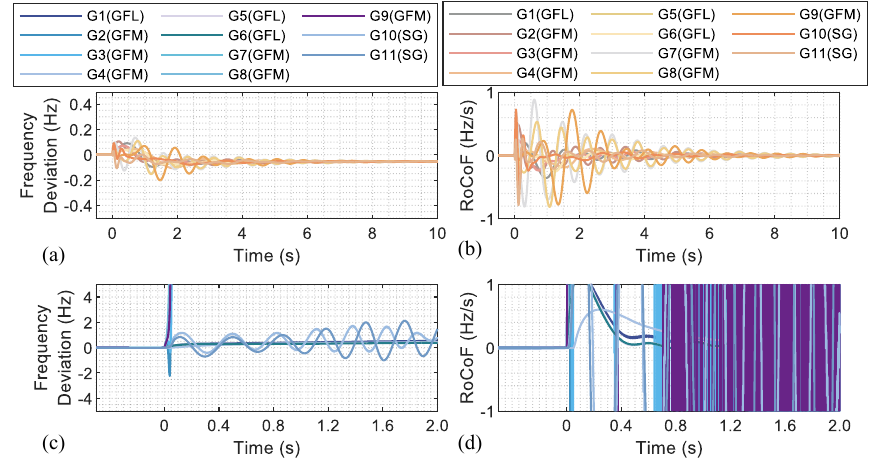}
\caption{ (a) Frequency deviations and (b) RoCoF of the system frequency under the solution provided by Scenario 1 following a 300 MW load disturbance. (c) Frequency deviations and (d) RoCoF after switching the inertia and damping to the solution computed in Scenario 2; the system cannot be stabilized under this operating point.}
\label{fig:normal}
\end{figure}

The time-domain simulations for Scenarios 1 and 2 are shown in Fig. \ref{fig:normal}(a-d). It can be observed that the system is small-signal unstable when we allocate inertia and damping according to the solution obtained in Scenario 2, as shown in Fig. \ref{fig:normal}(c-d). When we switch the inertia and damping of GFLs to the desired values, the system collapses. Some DERs lose synchronization and others begin oscillating. However, with the solution obtained from Scenario 1, the system can maintain small-signal stability at the operating point. Furthermore, with a 300 MW load disturbance, the induced frequency swing remains within the expected limits.

\subsection{Economic Analysis and Reserves}
\begin{figure}[t]
\centering
\includegraphics[width=0.49\textwidth]{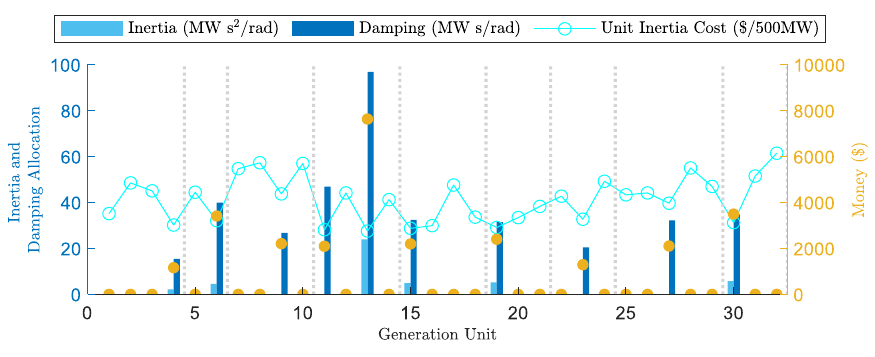}
\caption{The allocations of inertia and damping coefficients, together with the unit inertia prices, and the market payment for all individual DERs.}
\label{fig:allocation}
\end{figure}
Fig. \ref{fig:allocation} presents the allocations of inertia and damping coefficients for each unit, along with the corresponding market payments. The total payment to DERs is approximately $2.8024\times 10^4$ \$, while the aggregate cost is $2.4645\times 10^4$ \$. The higher payment to generators, compared to the cost, is due to the inherent budget imbalance characteristic of the VCG mechanism. Assuming an electricity price of $50$ \$/MWh, this payment corresponds to $550$ MWh of electricity. For comparison, the amount of electricity traded in the energy markets for this system is approximately $4000$ MWh. Establishing this market incurs an additional cost of approximately $14$\%. This market can incentivize DERs' contribution to system stability at an acceptable extra expenditure. 

Regarding DERs' power and energy reserve, we perform simulations by varying the location of the load disturbance and, based on the simulation results, calculate the maximum power rating that each DER must reserve. We also compute the energy change of each DER within 20 seconds after a fault. Scatter plots illustrating the relationship between the maximum power rating and the corresponding \(m\) and \(d\) under three different load cases are provided in Fig. \ref{fig:headroom}(a-c), and the specific reserved power values and energy losses are presented in Fig. \ref{fig:headroom}(d-e). These results indicate that while the exact reserved power varies with the disturbance location, a clear correlation exists between the reserved power and the parameters \(m\) and \(d\).

\begin{figure}[t]
\centering
\includegraphics[width=0.5\textwidth]{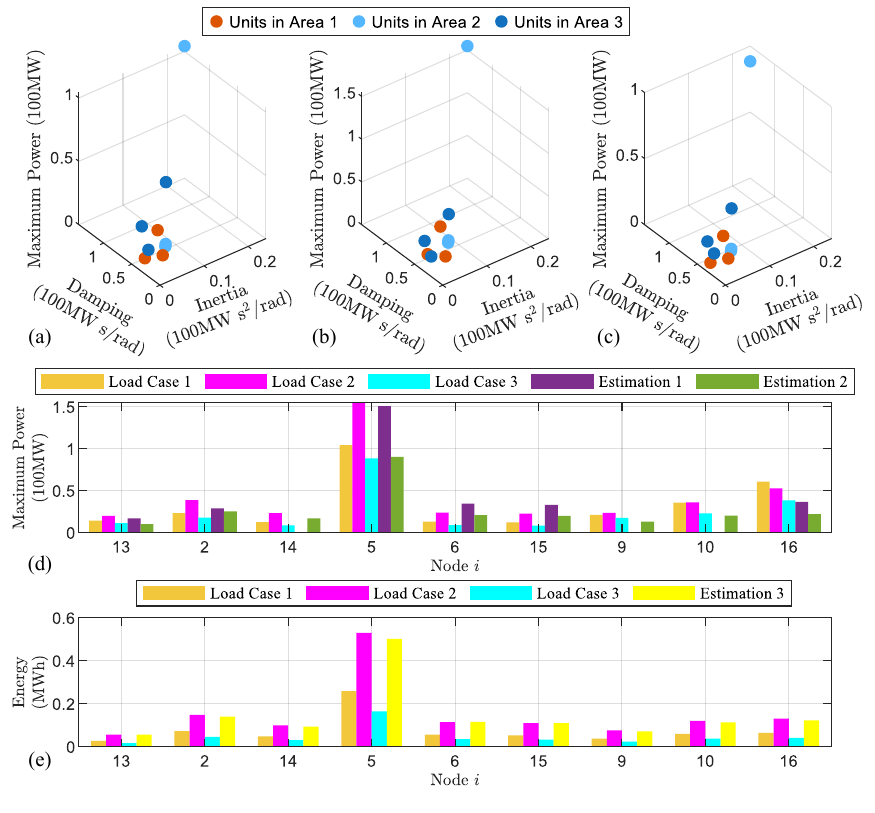}
\caption{(a–c) Scatter plots illustrating the relationship between the reserve DER power rating, nodal damping, and nodal inertia for load disturbances occurring at different locations, denoted as load cases 1–3, respectively. (d) The maximum reserve power rating for three different load cases and two estimation methods. (e) The energy loss for three different load cases and one estimation method.}
\label{fig:headroom}
\end{figure}

Variations in reserve power and energy are influenced by the network’s power redistribution and the differences in the controllers’ dynamics. For estimating reserve power ratings, one can directly use the product of the system operator’s RoCoF constraint value and the DER’s \(m\) (Estimation Method 1), or alternatively, the product of the system operator’s steady-state frequency deviation constraint value and the DER’s \(d\) (Estimation Method 2). For energy loss, one can estimate it by multiplying the system operator’s steady-state frequency deviation constraint value by the DER’s \(d\) and the duration of time (Estimation Method 3). Comparisons between these estimated values and the actual power and energy responses are shown in Fig. \ref{fig:headroom}(d-e). It can be seen that the estimated values generally resemble the actual reserve requirements under many conditions, while they do not fully guarantee robustness. Future research should focus on developing more rigorous methods to accurately and robustly evaluate the required reserve power rating and energy capacity.

\subsection{Load Dynamics' Effects}
We apply Kron reduction to simplify the network and load models, although loads can influence the system response. To investigate the effect of load characteristics, we consider four load scenarios: (1) all loads are constant power loads, (2) the majority of loads are constant current loads, (3) most loads are constant impedance loads, and (4) most loads are nonlinear loads with internal states, with parameters taken from \cite{choi2006measurement}. Detailed simulations of the system response under these load conditions are conducted. Fig. \ref{fig:load}(a–b) presents the frequency and RoCoF responses for the four cases.

The results indicate that different load characteristics indeed affect the system response. However, dynamic loads do not necessarily lead to poorer performance; for example, in the dynamic load scenario, both the oscillations and frequency deviations are smaller than those observed with constant power loads. Meanwhile, constant current loads may cause more serious oscillations. Overall, the frequency deviations and RoCoF approximately remain within acceptable limits, demonstrating that the proposed method is effective across a variety of practical scenarios.

\begin{figure}[t]
\centering
\includegraphics[width=0.5\textwidth]{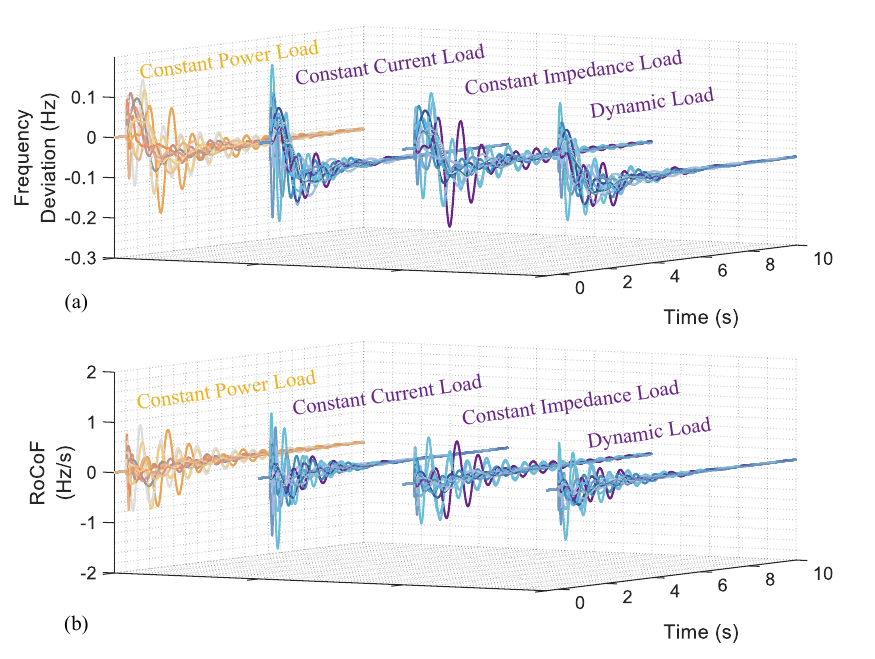}
\caption{Time domain simulations for (a) frequency deviations and (b) RoCoF for constant power loads, constant current loads, constant impedance loads, and dynamic loads simulation scenarios.}
\label{fig:load}
\end{figure}

\subsection{Realistic Consideration Test}
Simulation tests evaluate the robustness requirements. In our robust formulation, we assume that system parameter uncertainties are 10\% of their standard values. Additionally, we impose a damping uncertainty of 50 MW·rad/s for all GFL and SG nodes to prevent potential misestimations of damping values for these devices. We solve both the standard problem \eqref{eq:final_model} and the robust optimization problem \eqref{eq:robust_model}. The resulting allocations of damping and inertia are illustrated in Fig. \ref{fig:robust_allocation}. In both the standard and robust cases, the inertia allocation remains at 47.75 MW·s$^2$/rad; however, the total damping is 378.53 MW·s/rad in the non-robust case and increases to 490.08 MW·s/rad in the robust case. The robust constraints ensure that the system remains stable under various conditions, necessitating a higher damping allocation. Notably, the spatial allocation of damping and inertia also differs under robust conditions. In particular, the damping allocated to the unit at node 5 is significantly reduced, resulting in a more balanced distribution across regions.

\begin{figure}[t]
\centering
\includegraphics[width=0.5\textwidth]{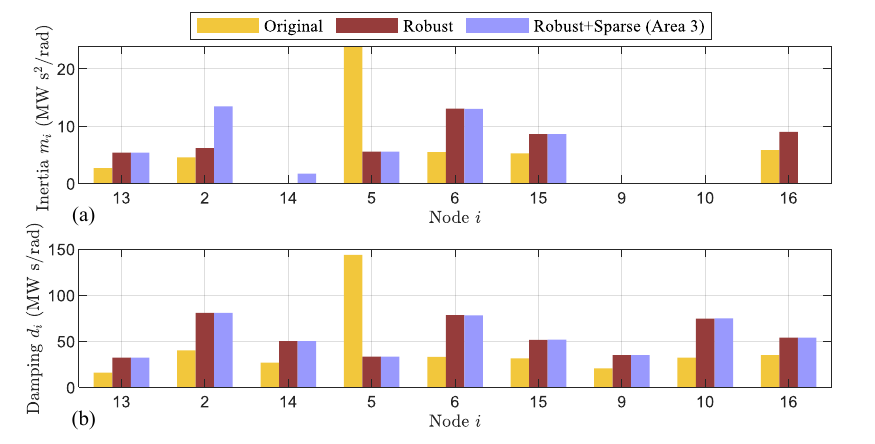}
\caption{The allocation of (a) inertia and (b) damping for standard problem, robust problem, and robust+sparse problem.}
\label{fig:robust_allocation}
\end{figure}

\begin{figure}[t]
\centering
\includegraphics[width=0.5\textwidth]{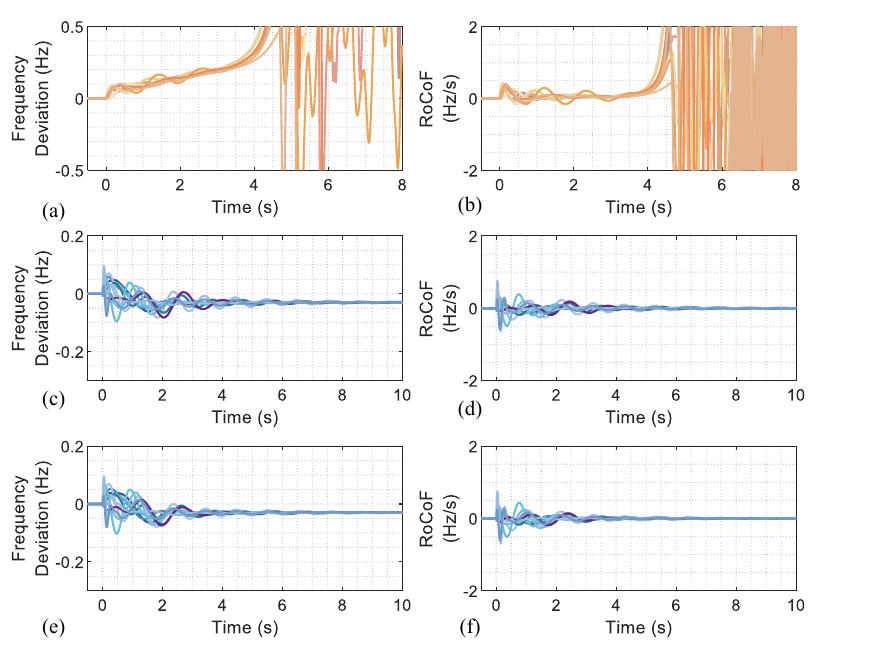}
\caption{Time-domain simulations for frequency deviations and RoCoF are presented for three cases: the standard problem (a–b), the robust problem (c–d), and the robust+sparse problem (e–f), under the condition that the network impedance is increased to 1.1 times its original value. In (a–b), the disturbance is introduced by adjusting the GFM damping to the desired value, and the system cannot be stabilized at the operating point. In (c–f), a 300 MW load disturbance is applied.}
\label{fig:robust}
\end{figure}

Furthermore, we test the solution's robustness by increasing the impedance of all transmission lines to 1.1 times their original values. The simulation results are shown in Fig. \ref{fig:robust}(a-d). With these modified system parameters, the original allocation scheme becomes unstable under the new operating conditions. Specifically, this operating point is unstable: when we set DER damping to the desired values, the frequency deviates from its nominal value, resulting in system instability and severe oscillations. In contrast, the solution obtained from the robust optimization problem enables the system to maintain stable operation with excellent dynamic responses, even when subjected to a 300 MW load disturbance, while keeping the RoCoF minimal.

We also test the effect of sparsity requirements. A sparsity-promoting penalty term is added to the inertia of all DERs in Region 3 (nodes 9, 10, and 16), using a penalty coefficient of $10^4$. The resulting inertia and damping allocations are shown in Fig. \ref{fig:robust_allocation}. The simulated time-domain responses are shown in Fig. \ref{fig:robust}(e-f). After accounting for sparsity, the overall system inertia and damping remain unchanged, albeit at a higher cost. At the same time, the system maintains stability even when network parameters vary. This demonstrates that our model effectively addresses practical factors such as sparsity and robustness. The extra cost associated with promoting sparsity represents an economic trade-off; much like in energy markets where some inexpensive units may be disqualified from participation, the system might incur additional costs.

\subsection{Scalability Test}
To assess the scalability of the proposed method, we select the WECC (263 nodes) \cite{zhang2023dataset}, South Carolina 500 (500 nodes) \cite{birchfield2016grid}, and Texas 2000 (2000 nodes) \cite{birchfield2016grid} test systems as examples of large-scale systems. First, we measure the time required to solve a single optimization problem for each system and find that the solution time is within 3 seconds; detailed data are provided in Tab. \ref{tab:solution_times}. We further validate the effectiveness of the stability constraints. The modal distributions corresponding to the optimized inertia and damping for the three systems are shown in Fig. \ref{fig:large}, demonstrating that the proposed formulation effectively restricts the system mode locations.

\begin{figure}[t]
\centering
\includegraphics[width=0.5\textwidth]{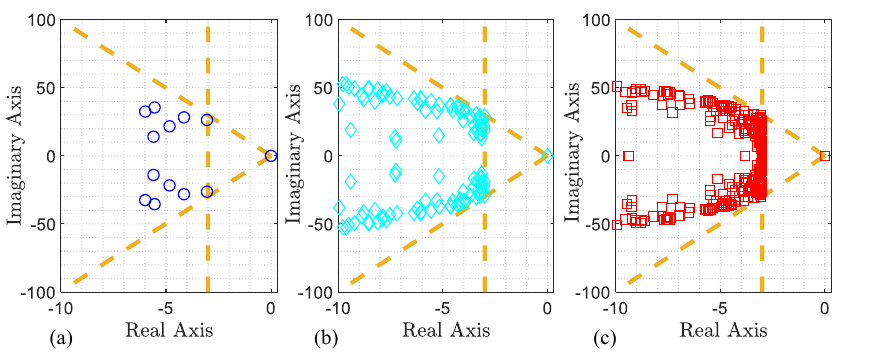}
\caption{The eigenvalue maps near the origin for (a) WECC test system; (b) South Carolina 500 test system; (c) Texas 2000 test system.}
\label{fig:large}
\end{figure}

\begin{table}[t]
\centering
\caption{Solution Times for Different Power Systems}
\resizebox{\columnwidth}{!}{%
\begin{tabular}{ccccc}
    \toprule
    & \textbf{Total Bus} & \textbf{Generator} & \textbf{Single Problem} & \textbf{Market Solution} \\
    \textbf{System} & \textbf{Number} & \textbf{Bus Number} & \textbf{Solution Time} & \textbf{Time} \\
    \midrule
    WECC & 263 & 29 & 0.34 s & 8.10 s \\
    South Carolina 500 & 500 & 90 & 1.15 s & 33.16 s \\
    Texas 2000 & 2000 & 282 & 2.38 s & 219.23 s \\
    \bottomrule
\end{tabular}
}
\label{tab:solution_times}
\end{table}

Regarding the market settlement process, the computational complexity increases linearly with the number of units because it requires computing each unit's contribution to reducing stability costs. By utilizing Matlab's parallel computing capabilities, payment amounts are computed efficiently. Detailed computational times are provided in Tab. \ref{tab:solution_times}, demonstrating that final results are obtained within 4 minutes. These outcomes indicate that the proposed approach is applicable to large-scale systems.

\subsection{Taylor Approximation Error Analysis}
In frequency stability analysis, our model uses a Taylor expansion to approximate the frequency nadir, which may introduce approximation errors relative to the original nadir values. To evaluate this, we conduct an error analysis of the effect of the Taylor approximation. We select two cases: one with fast and strong turbine dynamics (\( r^{-1}_i=5 \) MW$\cdot\mathrm{s}$/rad  and \( \tau_i=10 \) s), and another with slow and weak turbine dynamics (\( r^{-1}_i=10 \) MW$\cdot\mathrm{s}$/rad  and \( \tau_i=2 \) s). Simulations indicate that the COI frequency nadir value matches that computed by the nonlinear function in Ref. \cite{Global2020,guggilam2018optimizing}. Figs. \ref{fig:taylor}(a)(d) show the inverse of the nadir value as provided by the nonlinear nadir function and its Taylor approximation. Figs. \ref{fig:taylor}(b)(e) depict the curvature of the nonlinear function. Figs. \ref{fig:taylor}(c)(f) present the percentage error of the Taylor approximation compared with the original nonlinear function.

\begin{figure}[t]
\centering
\includegraphics[width=0.5\textwidth]{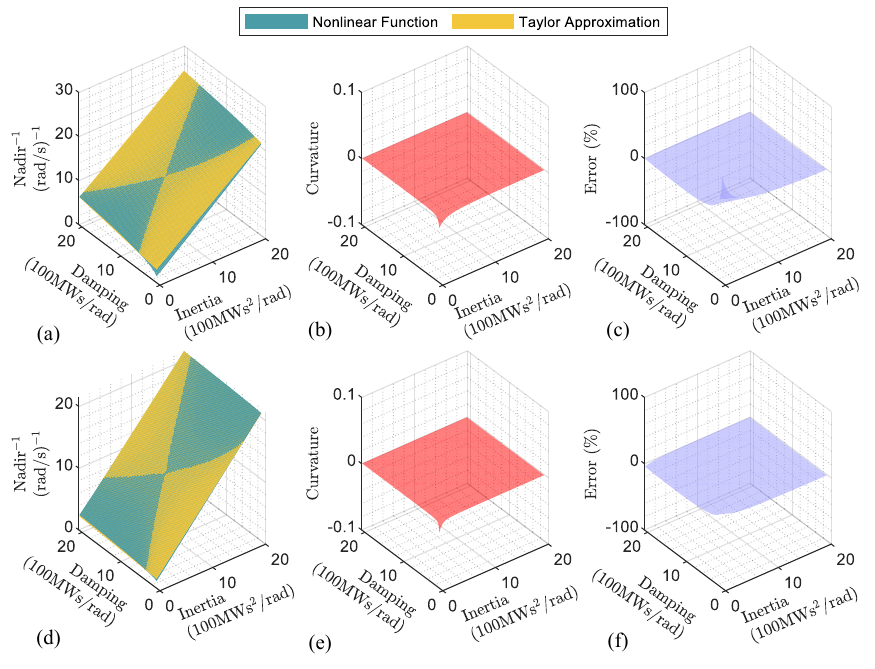}
\caption{(a–c) Inverse nadir value, nonlinear function curvature, and percentage error of the Taylor approximation for the strong and fast turbine case; (d–f) inverse nadir value, nonlinear function curvature, and percentage error of the Taylor approximation for the weak and slow turbine case.}
\label{fig:taylor}
\end{figure}

From the results, over most of the range of damping and inertia, the curvature of the nonlinear nadir function is nearly zero, indicating that it can be approximated by a flat plane. Consequently, in most cases, the Taylor approximation yields excellent results, with errors remaining within 3\%. Only when both inertia and damping are extremely low do nonlinear effects become particularly significant. Such a scenario is highly unlikely in real power systems, as it would imply that the system dynamics have virtually no buffering (almost no inertia) and no capability to resist disturbances (almost no damping). Thus, under typical operating conditions, the Taylor approximation is both accurate and acceptable.

\section{Conclusion}\label{conclusions}
This paper presents a hybrid oscillation damping and inertia management strategy for DERs enabled by economic incentives. The case studies demonstrate that implementing this management framework incurs minimal additional costs while effectively mitigating oscillations and ensuring frequency stability. To extend this research, several promising directions are proposed:

1. Enhanced Dynamic Modeling: Develop a more detailed dynamic model that incorporates diverse multi-loop dynamics of inverters and the behavior of other network components (e.g., fast load dynamics) under high DER penetration \cite{Hatziargyriou2021}. In addition, revising the classical Kron reduction method is important to better capture the dynamic behavior of load buses.

2. Nonlinear System Dynamics: In many cases, nonlinearities significantly influence system oscillation behaviors. Therefore, it is imperative to employ novel nonlinear analysis methods to accurately quantify these effects, identify oscillation sources, and mitigate them \cite{wang2019nonlinear,xu2020real}. For transient stability, DERs may exhibit strong coupling between frequency and voltage control. However, the existence of an energy function for grid-following inverters in a grid-connected (non-single-machine infinite-bus) system remains questionable \cite{zhang2025interaction}.

3. Adaptive Stability Contribution Evaluation: Numerous studies have explored the adaptive adjustment of inertia and damping parameters in real time, demonstrating excellent performance in mitigating oscillations \cite{li2016self,mir2019self,wang2022adaptive}. Evaluating the contribution of these adaptive DERs to overall system stability remains an open challenge. Furthermore, although we have modeled inertia and damping uncertainties using robust optimization techniques, future research should focus on bypassing parameter identification altogether by directly measuring each unit's contribution to system stability.

4. Accurate Reserve Estimation: Precise estimation of the curtailment amount is crucial for operational planning and market settlement. However, the dynamic nature of the power grid complicates accurate estimation. Future work should develop robust estimation techniques that account for grid complexities, including the effects of network topology and disturbance propagation.

\bibliographystyle{IEEEtran}
\bibliography{IEEEabrv,Reference}
\section*{Appendix}
\subsection{Proof for Proposition \ref{th:decay}}
Consider the right-shifting version of the eigenvalue solution of Eq.\eqref{eq:quadratic_eigenvalue}: 
\begin{align}
\,\,\,\,&(\hat{\lambda}-\beta )^2\mathbf{M}\boldsymbol{\nu }+(\hat{\lambda}-\beta )\mathbf{D}\boldsymbol{\nu }+\mathbf{L}\boldsymbol{\nu } \label{eq:shifted_system}
\\
=&\hat{\lambda}^2\mathbf{M}\boldsymbol{\nu }+\hat{\lambda}\left( \mathbf{D}-2\beta \mathbf{M} \right) \boldsymbol{\nu }+\left( \mathbf{L}-\beta \mathbf{D}+\beta ^2\mathbf{M} \right) \boldsymbol{\nu }=\mathbf{0}. \nonumber
\end{align}
where $\hat{\lambda}=\lambda+\beta$ is the solution to the shifted problem Eq.\eqref{eq:shifted_system} when $\lambda$ is the solution to Eq.\eqref{eq:quadratic_eigenvalue}. When all eigenvalues of the shifted system (excluding the eigenvalue $0+\beta$, cf. Lemma \ref{lemma:eig}) lie in the left half-plane, the dominant pole of the original system will reside in the plane $\mathrm{Re}\left( \lambda \right) \leqslant -\beta$.  

We will employ the subsequent definitions and lemmas. The dimension will be represented as \(N\). The orthogonal complement subspace of $\mathbf{1}$ in $\mathbb{R}^N$ is denoted by $\mathbb{I} _{\bot}=\left\{ \boldsymbol{x}\in \mathbb{R} ^N |\mathbf{1}^{\top}\boldsymbol{x}=0 \right\} $.

\begin{lemma}\label{symmetric}
Given a real symmetric matrix \(\mathbf{A}\), the eigenvectors associated with different eigenvalues are orthogonal.
\end{lemma}
\begin{lemma}[Subspace Positive Semi-definiteness] \label{orth}
If $\mathbf{L}-\beta \mathbf{D}+\beta ^2\mathbf{M}+v\mathbf{1}\mathbf{1}^{\top}\succcurlyeq 0
$, then the matrix $\mathbf{L}-\beta \mathbf{D}+\beta ^2\mathbf{M}$ is positive semi-definite on $\mathbb{I} _{\bot}$.
\end{lemma}
\begin{proof}
Because $\mathbf{L}-\beta \mathbf{D}+\beta ^2\mathbf{M}+v\mathbf{1}\mathbf{1}^{\top}\succcurlyeq 0$, we have $ \boldsymbol{x}^{\top}\left( \mathbf{L}-\beta \mathbf{D}+\beta ^2\mathbf{M}+v\mathbf{1}\mathbf{1}^{\top} \right) \boldsymbol{x}\geqslant 0$
for all $\boldsymbol{x}\in\mathbb{I} _{\bot}$. Besides, due to $\mathbf{1}^{\top}\boldsymbol{x}=0$, the equation above suggests $
\boldsymbol{x}^{\top}\left( \mathbf{L}-\beta \mathbf{D}+\beta ^2\mathbf{M} \right) \boldsymbol{x}\geqslant 0 $ for all $\boldsymbol{x}\in\mathbb{I} _{\bot}$.
\end{proof}
\begin{lemma}[Eigenvalues of Subspace Positive Semi-definite Matrix] \label{zero}
If a real symmetric matrix $\mathbf{A}$ is positive definite $\mathbf{A}\succ 0$ on $\mathbb{I} _{\bot}$, then $\mathbf{A}$ possesses at most one eigenvalue $\leqslant 0$.
\end{lemma}
\begin{proof}
We first prove that there is at most one $0$ eigenvalue:

Suppose a vector $\boldsymbol{x}\in\mathbb{I} _{\bot}$ is the solution to $\mathbf{A}\boldsymbol{x}=\mathbf{0}$, which leads to $\boldsymbol{x}^{\top}\mathbf{A}\boldsymbol{x}=0$, which is a contradiction with $\mathbf{A}\succ 0$ on $\mathbb{I} _{\bot}$. Consequently, all the vectors $\boldsymbol{x}\in\mathbb{I} _{\bot}$ cannot lie in the null space of $\mathbf{A}$. The dimension of the null space for $\mathbf{A}$ must be less than or equal to $1$: $\mathrm{dim}\left( \mathrm{null}\left( \mathbf{A} \right) \right) \leqslant N-\mathrm{dim}\left( \mathbb{I} _{\bot} \right) =1$. By the rank-nullity relationship, we know the rank of $\mathrm{rank}\left( \mathbf{A} \right) =N-\mathrm{dim}\left( \mathrm{null}\left( \mathbf{A} \right) \right) \geqslant N-1$, which indicates there are at least $N-1$ non-zero eigenvalues. 

Next, we prove there cannot be multiple non-positive eigenvalues: By contradiction, suppose there are two different eigenvalues $\sigma _1\leqslant0$ and $\sigma _2\leqslant0$, whose eigenvectors are $\boldsymbol{y}_1$ and $\boldsymbol{y}_2$, respectively. Given that $\mathrm{rank}\left( \mathbf{A} \right) \geqslant N-1$, $\sigma _1$ and $\sigma _2$ cannot be $0$ simultaneously. Since $\mathbf{A}$ is a real symmetric matrix, according to lemma \ref{symmetric}, $\boldsymbol{y}_1\bot \boldsymbol{y}_2$. Because of the linear independence of $\boldsymbol{y}_1$ and $ \boldsymbol{y}_2$, their span has at least a dimension of 2. We have $\mathrm{span}\left\{ \boldsymbol{y}_1,\boldsymbol{y}_2 \right\} \cap \mathbb{I} _{\bot}\ne \varnothing$.

Thereby, there is a nonzero real vector $\boldsymbol{z}\in \mathrm{span}\left\{ \boldsymbol{y}_1,\boldsymbol{y}_2 \right\} $ such that $\boldsymbol{z}\in\mathbb{I} _{\bot}$. Let $\boldsymbol{z}=a\boldsymbol{y}_1+b\boldsymbol{y}_2$ where at least one of $a,b$ is not zero. We have:
\begin{equation}
    \begin{aligned}
        &\boldsymbol{z}^{\top}\mathbf{A}\boldsymbol{z}=\left( a\boldsymbol{y}_1+b\boldsymbol{y}_2 \right) ^{\top}\mathbf{A}\left( a\boldsymbol{y}_1+b\boldsymbol{y}_2 \right) 
        \\
        &=a^2\boldsymbol{y}_{1}^{\top}\mathbf{A}\boldsymbol{y}_1+ab\boldsymbol{y}_{1}^{\top}\mathbf{A}\boldsymbol{y}_2+ab\boldsymbol{y}_{2}^{\top}\mathbf{A}\boldsymbol{y}_1+b^2\boldsymbol{y}_{2}^{\top}\mathbf{A}\boldsymbol{y}_2
        \\
        &=a^2\sigma _1\boldsymbol{y}_{1}^{\top}\boldsymbol{y}_1+ab\sigma _2\boldsymbol{y}_{1}^{\top}\boldsymbol{y}_2+ab\sigma _1\boldsymbol{y}_{2}^{\top}\boldsymbol{y}_1+b^2\sigma _2\boldsymbol{y}_{2}^{\top}\boldsymbol{y}_2
        \\
        &=\sigma _1a^2\left\| \boldsymbol{y}_1 \right\| ^2+\sigma _2b^2\left\| \boldsymbol{y}_2 \right\| ^2\leqslant0.
    \end{aligned}
\end{equation}
which leads to a contradiction to the condition of $\boldsymbol{x}^{\top}\mathbf{A}\boldsymbol{x}> 0$ for all vectors $\boldsymbol{x}\in\mathbb{I} _{\bot}$. 
\end{proof}

Now we are ready to prove the Proposition 1.
\begin{proof}
The shifted eigenvalue solution problem is represented as follows:
\begin{equation}\label{eq:shifted_system1}
    \hat{\lambda}^2\mathbf{M}\boldsymbol{\nu }+\hat{\lambda}\left( \mathbf{D}-2\beta \mathbf{M} \right) \boldsymbol{\nu }+\left( \mathbf{L}-\beta \mathbf{D}+\beta ^2\mathbf{M} \right) \boldsymbol{\nu } = \boldsymbol{0}.
\end{equation}
The quadratic equation about $\hat{\lambda}$ can be expressed as:
\begin{equation} \label{eq:quadratic_equation1}
    \begin{aligned}
        &\hat{\lambda}^2\left< \boldsymbol{\nu },\mathbf{M}\boldsymbol{\nu } \right> +\hat{\lambda}\left< \boldsymbol{\nu },\left( \mathbf{D}-2\beta \mathbf{M} \right) \boldsymbol{\nu } \right>\\
        &+\left< \boldsymbol{\nu },\left( \mathbf{L}-\beta \mathbf{D}+\beta ^2\mathbf{M} \right) \boldsymbol{\nu } \right> =0.
    \end{aligned}
\end{equation}

For complex eigenvalue pairs $\hat{\lambda}$ and its conjugate $\hat{\lambda}^*$, notice the two complex roots of \eqref{eq:quadratic_equation1} satisfy:
\begin{align}
    \hat{\lambda}+\hat{\lambda}^*=\frac{-\boldsymbol{\nu }^{\mathrm{H}}\left( \mathbf{D}-2\beta \mathbf{M} \right) \boldsymbol{\nu }}{\boldsymbol{\nu }^{\mathrm{H}}\mathbf{M}\boldsymbol{\nu}}.
\end{align}
Thereby, if the first inequality of this proposition is satisfied, we have
\begin{equation} 
    \hat{\lambda}+\hat{\lambda}^*=\frac{-\boldsymbol{\nu }^{\mathrm{H}}\left( \mathbf{D}-2\beta \mathbf{M} \right) \boldsymbol{\nu }}{\boldsymbol{\nu }^{\mathrm{H}}\mathbf{M}\boldsymbol{\nu }}\leqslant 0, \label{eq:root_1}
\end{equation}
which further indicates the real part of $\hat{\lambda}$ is less than or equal to zero, and is sufficient to guarantee the complex eigenvalue pairs both have non-positive real parts.

For a real eigenvalue $\hat{\lambda}$ other than eigenvalue $0+\beta$, by contradiction, suppose there is a real eigenvalue larger than zero $0<\hat{\lambda}<\beta $ (it must hold that $\hat{\lambda}< 0+\beta$ because all the eigenvalues in the original system strictly lie in the left half plane). Notice \eqref{eq:shifted_system1} can be written as:
\begin{equation} \label{eq:linear}
    \underset{\mathbf{A}\left( \hat{\lambda} \right)}{\underbrace{\left( \hat{\lambda}^{2}\mathbf{M}+\hat{\lambda}\left( \mathbf{D}-2\beta \mathbf{M} \right) +\left( \mathbf{L}-\beta \mathbf{D}+\beta ^2\mathbf{M} \right) \right) }}\boldsymbol{\nu }=\mathbf{0}.
\end{equation}

Given that \(\hat{\lambda}>0\), the following relations hold: \(\hat{\lambda}^2\mathbf{M}\succ 0\) and 	\(\hat{\lambda}\left( \mathbf{D}-2\beta \mathbf{M} \right) \succcurlyeq 0\), which is valid because \(\mathbf{D}-2\beta \mathbf{M}\succcurlyeq 0\). Based on Lemma \ref{orth}, the matrix \(\mathbf{L}-\beta \mathbf{D}+\beta ^2\mathbf{M}\succcurlyeq 0\) is necessarily positive semi-definite over \(\mathbb{I} _{\bot}\). When considering the summation of these matrices, the matrix \(\mathbf{A}( \hat{\lambda} )\) must be positive definite on \(\mathbb{I} _{\bot}\). Referring to Lemma \ref{zero}, this matrix can have, at most, a single eigenvalue \(\leqslant 0\). The other $N-1$ eigenvalues must be positive. Our subsequent task is to establish that this eigenvalue must be negative rather than zero.

For $\mathbf{A}( \hat{\lambda} )$, we have:
\begin{align}
    \mathbf{A}( \hat{\lambda} ) &=(\hat{\lambda}-\beta )^2\mathbf{M}+(\hat{\lambda}-\beta )\mathbf{D}+\mathbf{L}
    \\
    &\preccurlyeq \frac{(\hat{\lambda}-\beta )^2}{2\beta}\mathbf{D}+(\hat{\lambda}-\beta )\mathbf{D}+\mathbf{L} \label{eq:less}
    \\
    &=\frac{(\hat{\lambda}-\beta )(\hat{\lambda}+\beta )}{2\beta}\mathbf{D}+\mathbf{L} 
    \\
    &\prec \mathbf{L}, \label{eq:mono}
\end{align}
where \eqref{eq:less} is due to $\mathbf{D}-2\beta \mathbf{M}\succcurlyeq 0$; \eqref{eq:mono} is due to $\hat{\lambda}-\beta<0$. The above inequality implies:
\begin{equation}
    \mathbf{1}^{\top}\mathbf{A}( \hat{\lambda} ) \mathbf{1} < \mathbf{1}^{\top}\mathbf{L}\mathbf{1} = 0.
\end{equation}

From this, it's evident that the matrix \(\mathbf{A}( \hat{\lambda} )\) has at least one negative eigenvalue. Consequently, the last eigenvalue must be negative. Given \(\mathbf{A}( \hat{\lambda} )\) possesses one negative eigenvalue and \(N-1\) positive eigenvalues, \(\mathbf{A}( \hat{\lambda} )\) is full rank. The only feasible solution to the equation \(\mathbf{A}( \hat{\lambda} ) \boldsymbol{\nu} = \mathbf{0}\) is \(\boldsymbol{\nu} = \mathbf{A}( \hat{\lambda} )^{-1}\mathbf{0} = \mathbf{0}\). This presents an inconsistency, implying that no real solution exists in the range \(0 < \hat{\lambda} < \beta\) to solve the quadratic eigenvalue problem. 
\end{proof}

\subsection{Proof for Proposition \ref{th:conic}}
\begin{proof}
By multiplying both sides of Eq.\eqref{eq:quadratic_eigenvalue} with the conjugate transpose of eigenvector $\boldsymbol{\nu }^{\mathrm{H}}$, we can transform it into a scalar quadratic equation concerning $\lambda$:
\begin{equation} \label{eq:quadratic_equation}
    \lambda ^2\left< \boldsymbol{\nu },\mathbf{M}\boldsymbol{\nu } \right> +\lambda \left< \boldsymbol{\nu },\mathbf{D}\boldsymbol{\nu } \right> +\left< \boldsymbol{\nu },\mathbf{L}\boldsymbol{\nu } \right> =0.
\end{equation}
For real eigenvalue pairs, inequality \eqref{eq:conic_original} is naturally satisfied because the discriminant of \eqref{eq:quadratic_equation} $\left< \boldsymbol{\nu },\mathbf{D}\boldsymbol{\nu } \right> ^2-4\left< \boldsymbol{\nu },\mathbf{M}\boldsymbol{\nu } \right> \left< \boldsymbol{\nu },\mathbf{L}\boldsymbol{\nu } \right> \geqslant 0$.
For complex conjugate eigenvalue pairs $\lambda$ and $\lambda ^*$, the necessary and sufficient condition for location in $\sin \zeta \cdot \mathrm{Re}\left( \lambda \right) +\cos \zeta \cdot \mathrm{Im}\left( \lambda \right) \leqslant 0$ is:
\begin{equation} \label{eq:prop2}
    \left( \sin \zeta \right) ^2\left( \lambda +\lambda ^* \right) ^2+\left( \cos \zeta \right) ^2\left( \lambda -\lambda ^* \right) ^2\geqslant 0.
\end{equation}
Applying Vieta's formulas to Eq.\eqref{eq:quadratic_equation}:
\begin{align}
    &\left( \lambda +\lambda ^* \right) ^2=\left( \frac{-\left< \boldsymbol{\nu },\mathbf{D}\boldsymbol{\nu } \right>}{\left< \boldsymbol{\nu },\mathbf{M}\boldsymbol{\nu } \right>} \right) ^2,
    \\
    & \left( \lambda -\lambda ^* \right) ^2=\left( \frac{-\left< \boldsymbol{\nu },\mathbf{D}\boldsymbol{\nu } \right>}{\left< \boldsymbol{\nu },\mathbf{M}\boldsymbol{\nu } \right>} \right) ^2-4\frac{\left< \boldsymbol{\nu },\mathbf{L}\boldsymbol{\nu } \right>}{\left< \boldsymbol{\nu },\mathbf{M}\boldsymbol{\nu } \right>}.
\end{align}
Taking them into the \eqref{eq:prop2}, the proof is completed.
\end{proof}

\subsection{Proof for Theorem \ref{th:placement}}
\begin{proof}
From the first inequality of Proposition 1, we have:
\begin{equation}
    \left< \boldsymbol{\nu },\mathbf{D}\boldsymbol{\nu } \right> \geqslant 2\beta \left< \boldsymbol{\nu },\mathbf{M}\boldsymbol{\nu } \right> > 0,
\end{equation}
As a result, we obtain:
\begin{equation}
    \left< \boldsymbol{\nu },\mathbf{D}\boldsymbol{\nu } \right> ^2\geqslant 2\beta \left< \boldsymbol{\nu },\mathbf{M}\boldsymbol{\nu } \right> \left< \boldsymbol{\nu },\mathbf{D}\boldsymbol{\nu } \right> .
\end{equation}
Therefore, in order to make the constraint of Proposition 2 hold, the following condition is sufficient:
\begin{equation} \label{eq:th1-2}
    \begin{aligned}
        2\beta \left< \boldsymbol{\nu },\mathbf{M}\boldsymbol{\nu } \right> \left< \boldsymbol{\nu },\mathbf{D}\boldsymbol{\nu } \right> -4\left( \cos \zeta \right) ^2\left< \boldsymbol{\nu },\mathbf{L}\boldsymbol{\nu } \right> \left< \boldsymbol{\nu },\mathbf{M}\boldsymbol{\nu } \right> \geqslant 0.
    \end{aligned}
\end{equation}
Because $\left< \boldsymbol{\nu },\mathbf{M}\boldsymbol{\nu } \right>>0$, $\beta\mathbf{D}-2\left( \cos \zeta \right) ^2\mathbf{L}\succcurlyeq 0$ guarantees \eqref{eq:th1-2} and the proof is completed.
\end{proof}

\begin{IEEEbiography}[{\includegraphics[width=1in,height=1.25in,clip,keepaspectratio]{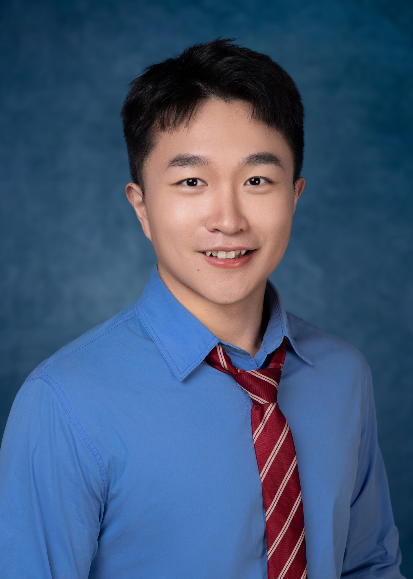}}]{Cheng Feng} (Member, IEEE) received the B.S. degree in Electrical Engineering from Huazhong University of Science and Technology in June 2019, and the Ph.D. degree in Electrical Engineering from Tsinghua University in June 2024. During February 2023 to August 2023, he was a visiting scholar in the Automatic Control Laboratory, ETH Zürich. He is now the Ezra Postdoctoral Associate at Cornell University. His research interests include virtual power plants, power system flexibility, and stability.
\end{IEEEbiography}

\begin{IEEEbiography}[{\includegraphics[width=1in,height=1.25in,clip,keepaspectratio]{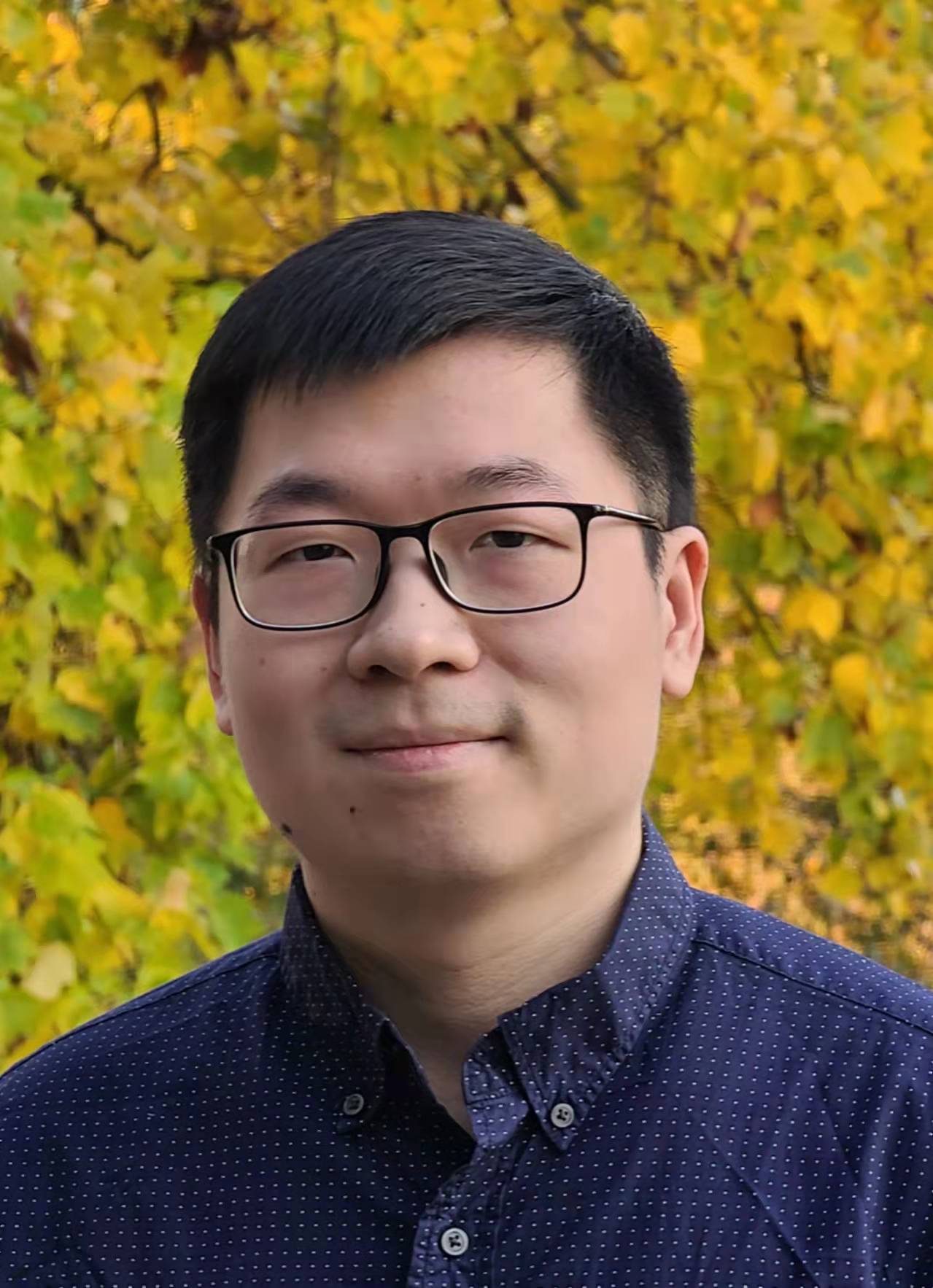}}] {Linbin Huang} (Member, IEEE) received the B.Eng. and Ph.D. degrees from Zhejiang University, Hangzhou, China, in 2015 and 2020, respectively. Currently, he is a ZJU100 Young Professor at Zhejiang University, Hangzhou, China. From 2020 to 2024, he was with ETH Zürich, Switzerland, where he was a Postdoctoral Researcher and then promoted to Senior Scientist. His research interests include power system stability, optimal control of power electronics, and data-driven control.
\end{IEEEbiography}

\begin{IEEEbiography}[{\includegraphics[width=1in,height=1.25in,clip,keepaspectratio]{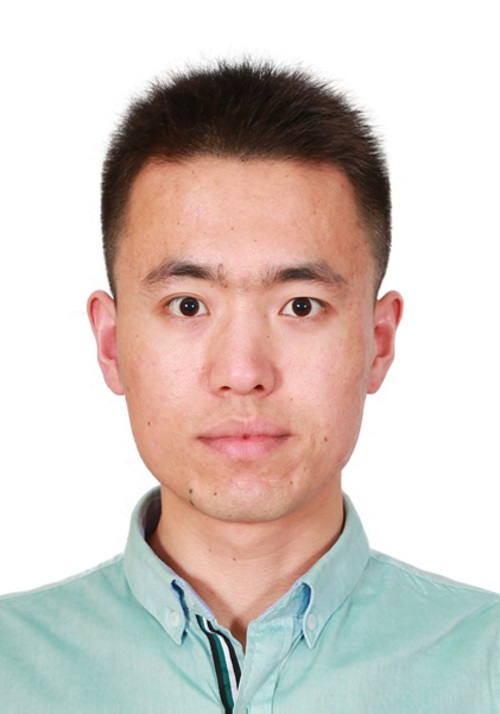}}] {Xiuqiang He} (Member, IEEE) received his B.S. degree and Ph.D. degree in control science and engineering from Tsinghua University, China, in 2016 and 2021, respectively. Since 2021, he has been a Postdoctoral Researcher and was later promoted to a Senior Scientist at ETH Zürich, Switzerland. His current research interests include power system dynamics, stability, and control. Dr. He was the recipient of the Beijing Outstanding Graduates Award and the Outstanding Doctoral Dissertation Award from Tsinghua University.
\end{IEEEbiography}

\begin{IEEEbiography}[{\includegraphics[width=1in,height=1.264in,clip,keepaspectratio]{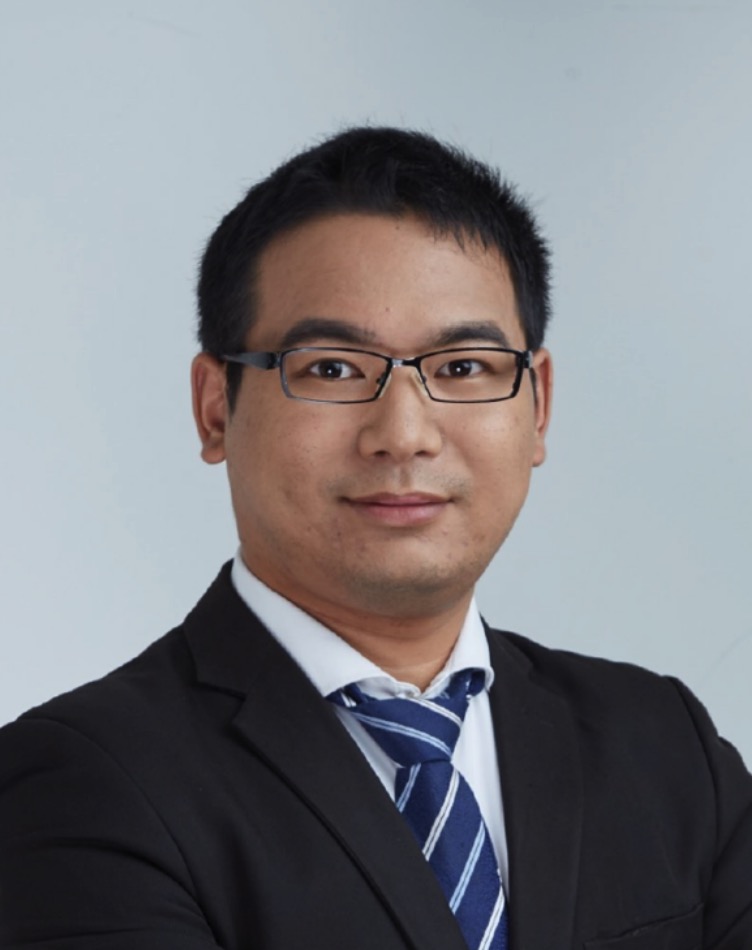}}]{Yi Wang} (Member, IEEE) received the B.S. degree from Huazhong University of Science and Technology in June 2014, and the Ph.D. degree from Tsinghua University in January 2019. He was a visiting student at the University of Washington from March 2017 to April 2018. He served as a Postdoctoral Researcher in the Power Systems Laboratory at ETH Zurich from February 2019 to August 2021. He is currently an Assistant Professor with the Department of Electrical and Electronic Engineering, The University of Hong Kong. His research interests include data analytics in smart grids, energy forecasting, multi-energy systems, Internet of Things, and cyber-physical-social energy systems.
\end{IEEEbiography}

\begin{IEEEbiography}[{\includegraphics[width=1in,height=1.264in,clip,keepaspectratio]{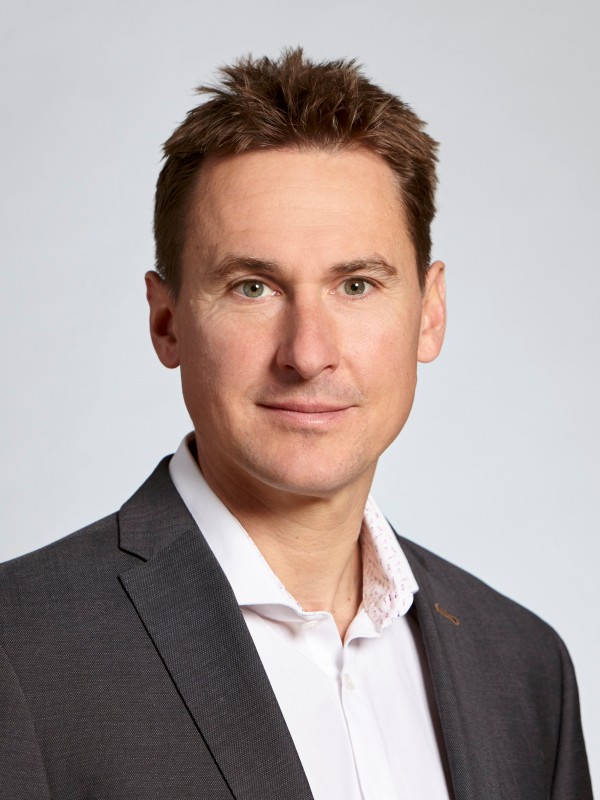}}]{Florian Dörfler} (Senior Member, IEEE) received the diploma degree in engineering cybernetics from the University of Stuttgart, Stuttgart, Germany, in 2008, and the Ph.D. degree in mechanical engineering from the University of California, Santa Barbara, CA, USA, in 2013. He is currently a Full Professor with the Automatic Control Laboratory, ETH Zürich, Zürich, Switzerland. From 2013 to 2014, he was an Assistant Professor with the University of California, Los Angeles, CA, USA. From 2021 to 2022, he has been the Associate Head with the Department of Information Technology and Electrical Engineering, ETH Zürich. His research interests include control, optimization, and system theory with applications in network systems, in particular, electric power grids. 
\end{IEEEbiography}

\begin{IEEEbiography}[{\includegraphics[width=1in,height=1.264in,clip,keepaspectratio]{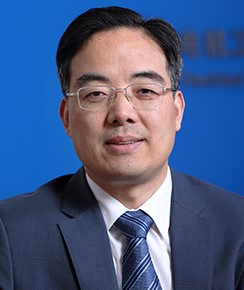}}]{Chongqing Kang} (Fellow, IEEE) received the Ph.D. degree from the Department of Electrical Engineering, Tsinghua University, Beijing, China, in 1997, where he is currently a Professor. His research interests include power system planning, power system operation, renewable energy, low-carbon electricity technology, and load forecasting.
\end{IEEEbiography}

\ifCLASSOPTIONcaptionsoff
\newpage
\fi
\end{document}